\documentclass[12pt,letterpaper,fleqn]{article}
\usepackage{etex,authblk,amsmath,fullpage,amsfonts,titletoc,amssymb,theorem}
\usepackage[onehalfspacing]{setspace}
\usepackage{graphicx,hyperref,pgfplots,dcolumn,ctable,floatpag}

\hypersetup{colorlinks=true,
            linkcolor=blue,
            citecolor=blue,
            urlcolor=blue}

\usetikzlibrary{pgfplots.groupplots}
\usetikzlibrary{patterns, arrows}
\usetikzlibrary{arrows,positioning,shadows,shapes,calc,fit,backgrounds}
\tikzset{>=stealth}

\reversemarginpar

\floatpagestyle{empty}

\makeatletter
\def\pgfplots@drawtickgridlines@INSTALLCLIP@onorientedsurf#1{}
\makeatother

\usepackage[font=footnotesize,labelfont=bf,labelsep=period,justification=justified]{caption}

\usepackage{natbib}
         \makeatletter
          \renewcommand{\bibsection}{
          \begin{center}
  \section*{\refname\@mkboth{\MakeUppercase{\refname}}
  {\MakeUppercase{\refname}}}
              \end{center}
              }
          \makeatother

\newtheorem{theorem}{Theorem}

\newtheorem{assumption}{Assumption}

\newtheorem{corollary}[theorem]{Corollary}

\newtheorem{definition}[theorem]{Definition}
\newtheorem{example}[theorem]{Example}

\newtheorem{lemma}[theorem]{Lemma}

\newtheorem{proposition}[theorem]{Proposition}
\newtheorem{theorem-app}{Theorem}[section]
\newtheorem{lemma-app}[theorem-app]{Lemma}
\newtheorem{proposition-app}[theorem-app]{Proposition}
\theorembodyfont{\upshape}

\newenvironment{proof}[1][\proofname]{
                \par\normalfont\trivlist\item[\hskip\labelsep\textbf{#1}.]\ignorespaces}
                {\hfill Q.E.D.\endtrivlist}
\newcommand{\proofname}{Proof}

\newcommand{\gD}{\Delta}
\newcommand{\ga}{\alpha}

\newcommand{\gl}{\lambda}
\newcommand{\eps}{\varepsilon}
\newcommand{\gd}{\delta}

\newcommand{\R}{{\mathbb R}}

\newcommand{\N}{{\mathbb N}}

\newcommand{\Z}{{\mathbb Z}}

\newcommand{\gO}{\Omega}
\newcommand{\go}{\omega}

\newcommand{\diag}{\operatorname{diag}}

\newcommand{\cI}{{\mathcal I}}

\newcommand{\cE}{{\mathcal E}}

\newcommand{\cP}{\mathcal{P}}
\newcommand{\cD}{\mathcal D}
\newcommand{\cR}{\mathcal{R}}

\newcommand{\cH}{{\mathcal H}}
\newcommand{\cA}{{\mathcal A}}
\newcommand{\cG}{{\mathcal G}}
\newcommand{\cK}{{\mathcal K}}
\newcommand{\cB}{{\mathcal B}}

\newcommand{\Supp}{{\operatorname{Supp}}}

\usepackage{bbm}
\usepackage{color}
\long\def\cmt#1{{}}

\usepackage{mathtools}
\mathtoolsset{showonlyrefs}

\newcommand{\be}{\begin{equation}}

\newcommand{\ee}{\end{equation}}
\newcommand{\bea}{\begin{eqnarray}}
\newcommand{\eea}{\end{eqnarray}}
\newcommand{\bee}{\begin{equation*}}
\newcommand{\eee}{\end{equation*}}

\begin{document}


\title{Optimal Transport of Information\thanks{We thank Darrell Duffie, Piotr Dworczak, Bart Lipman,  Jean-Charles Rochet, and Stephen Morris (AEA discussant) as well as seminar participants at Caltech, UBC, SUFE, SFI and conference participants at the 2020 AEA meeting in San Diego  for helpful comments. Parts of this paper were written when Malamud visited the Bank for International Settlements (BIS) as a research fellow. The views in this article are those of the authors and do not necessarily represent those of BIS.}}

\author{Semyon Malamud\thanks{%
Swiss Finance Institute, EPF Lausanne, and CEPR; E-mail: \texttt{semyon.malamud@epfl.ch%
}}, Anna Cieslak\thanks{Duke University, Fuqua School of Business, CEPR and NBER, E-mail: \texttt{anna.cieslak@duke.edu}}, and Andreas Schrimpf\thanks{Bank of International Settlements (BIS) and CEPR; E-mail: \texttt{andreas.schrimpf@bis.org%
} }\\
 }

\date{This version: \today}

\maketitle


\begin{abstract} \noindent  We study the general problem of Bayesian persuasion (optimal information design) with continuous actions and continuous state space in arbitrary dimensions. First, we show that with a finite signal space, the optimal information design is always given by a partition. Second, we take the limit of an infinite signal space and characterize the solution in terms of a Monge-Kantorovich optimal transport problem with an endogenous information transport cost. We use our novel approach to:
1. Derive necessary and sufficient conditions for optimality based on Bregman divergences for non-convex functions.
2. Compute exact bounds for the Hausdorff dimension of the support of an optimal policy.
3. Derive a non-linear, second-order partial differential equation whose solutions correspond to regular optimal policies. 
We illustrate the power of our approach by providing explicit solutions to several non-linear, multidimensional Bayesian persuasion problems. 

\vspace{5pt}

\vspace{5pt}
\end{abstract}

\renewcommand{\thefootnote}{\number\value{footnote}}

\pagenumbering{arabic}
\def\baselinestretch{1.617}\small\normalsize%

\clearpage

\section{Introduction}

We study the general problem of Bayesian persuasion (optimal information design) introduced in the seminal work of \cite{KamGenz2011}.\footnote{See also \cite{aumann1995}, \cite{pavan2006}, \cite{ostrovsky2010} and \cite{RayoSegal2010} for important prior contributions to the literature on communication with commitment. The term ``information design" was introduced in \cite{RePEc:edn:esedps:256} and \cite{bergemann2016information}. See \cite{BergMor2017} and \cite{kamenica2019bayesian} for excellent reviews.}  We show that solutions to optimal information design problem exhibit a remarkable mathematical structure when the private information of the sender (the state) is a random vector in $\R^L$ with a prior distribution absolutely continuous with respect to the Lebesgue measure. 

We start by solving a restricted problem in which the sender is constrained to a finite set of signals. In this case, we show that the optimal signal is always given by a partition of the state space. When the receivers' actions are functions of the expected state, the partition is given by convex polygons, being a natural analog of a monotone partition in many dimensions.\footnote{\cite{KleinbergMullainathan2019} argue that clustering (partitioning the state space into discrete cells) is the most natural way to simplify information processing in complex environments. Our results provide a theoretical foundation for such clustering. Note that, formally, in a Bayesian persuasion framework, economic agents (signal receivers) would need to use (potentially complex) calculations underlying the Bayes rule to compute the conditional probabilities. One important real-world problem arises when the receivers do not know the "true" probability distribution, in which case methods from robust optimization need to be used. See \cite{DworzakPavan2020}.}

Our explicit characterization of optimal partitions allows us to take the continuous limit and show that these partitions converge to a solution to the unconstrained problem. We establish a surprising connection between optimal information design and the Monge-Kantorovich theory of optimal transport whereby the sender effectively finds an optimal way of "transporting information" to the receiver, with an endogenous {\it information transport cost}. 

In the case when public actions are a function of expectations about (multiple, arbitrary) functions of the state (the moment persuasion, see \cite{DworzakKolotilin2019}), the information transport cost coincides with the classical Bregman divergence function (see \cite{rockafellar1970convex}), albeit the function involved is not convex, and hence the Bregman divergence loses all of its classic properties. We derive explicit necessary and sufficient conditions characterizing optimal policies that reduce to a system of non-linear partial (integro-)differential equations under enough regularity. We illustrate the power of this approach by providing explicit solutions to several non-linear, multidimensional information design problems. 

Unfortunately, in general, the smoothness of optimal policies cannot be guaranteed, and understanding their structure is very difficult. We use the theory of Hausdorff dimension in metric spaces and derive sharp upper bounds on the Hausdorff dimension of the support of the optimal policy, linking it explicitly to the "degree of convexity" of the underlying problem. 
 
\vspace{5pt}

Most existing papers on Bayesian persuasion with continuous signals consider the case of a one-dimensional signal space. See, e.g., \cite{gentzkow2016rothschild}, \cite{Kolotilin2018}, \cite{Hopenhayn2019}, \cite{DworczakMartini2019}, \cite{Arielietal2020} and \cite{kleiner2020extreme}. However, little is known about optimal information design in the multidimensional case. The only results we are aware of are due to \cite{DworczakMartini2019}.  They show that, under some technical conditions, with a two-dimensional state, four possible actions, and sender's utility that only depends on the expected state, the optimal information design is given by a partition into four convex polygons. See; also, \cite{DworzakKolotilin2019} who derive criteria for the support of the optimal policy to lie on a line. By contrast, we characterize the general solution to the problem of moment persuasion in any dimension, opening up a road to numerous real-world applications (see, e.g., \cite{das2017reducing}). Our characterization implies a certain type of monotonic association between the actual state and the optimal signal. This result could be viewed as a multidimensional extension of the results in \cite{kolotilin2020assortative}. 

Our paper is also related to the growing literature on martingale optimal transport in mathematical finance, the classic Monge-Kantorovich optimal transport under the constraint that the target is a conditional expectation of the origin. See, \cite{beiglbock2016problem} and \cite{ghoussoub2019structure}. 

Some of the results in our paper have been established in \cite{rochet1994insider} and \cite{kramkov2019optimal} for the special case of a two-dimensional first moment persuasion with a quadratic cost function. \cite{kramkov2019optimal} were the first to show a connection between such problems and optimal transport. As \cite{rochet1994insider} show, such problems are also related to a special class of signaling games.

A key technical contribution that makes our analysis possible is the solution to the finite signal case. The fact that optimal design is always a partition satisfying certain remarkable optimality conditions allows us to show that the full solution inherits the same properties. Heuristically, partition structure implies that the problem has a sufficiently regular solution (a map), and supports of maps are (in some sense) continuous in the limit. By contrast, with randomization, the optimal policy is a measure, and supports of measures are not continuous in the limit. In turn, quite remarkably, the proof of the partition result is based on the theory of real analytic functions. To the best of our knowledge, such techniques have never been used in optimal transport theory before. 

Finally, we note that our solution to the problem with a finite signal space relates this paper to the literature on optimal rating design (see, e.g., \cite{Hopenhayn2019}). Indeed, in practice, most ratings are discrete. For example, credit rating agencies use discrete rating buckets (e.g., above BBB-); restaurant and hotel ratings take a finite number of values. 

\section{The Model}

There are four time periods, $t\ =\ 0-,\ 0,\ 0+,\ 1.$ The information designer (the sender) believes that the state $\go$ (the private information of the sender) is a random vector taking values in  $\gO\subset \R^L$, an open subset of $\R^L$, and distributed with a density $\mu_0(\go)$ that is strictly positive on $\gO.$ There are $N\ge 1$ agents (receivers) who share the same prior $\mu_0(\go).$ 

Following \cite{KamGenz2011}, we assume that the sender is able to \emph{commit} to an information design at time $t=0-,$ before the state $\go$ is realized. 
The sender learns the realization of the state $\go$ at time $t=0,$ while the receivers only learn it at time $t=1.$ The sender's objective is then to decide how much, and what kind of, information about $\go$ to reveal to the receivers at time $t=0$.

\begin{definition}[Finite Information Design] An information design is a probability space $\cK$ (hereinafter signal space) and a probability measure $\cP$ on $\cK\times \gO.$ An information design is $K$-finite if the signal space $\cK$ has exactly $K$ elements: $|\cK|=K.$ An information design is finite if it is $K$-finite for some $K\in\N.$ In this case, without loss of generality we assume that $\cK\ =\ \{1,\cdots,K\}.$
\end{definition}

Once the receivers observe a signal $k\in \{1,\cdots,K\}$, they update their beliefs about the probability distribution of $\go$ using the Bayes rule. To do this, the receivers just need to know $\pi_k(\go)$ -- the probability of the state being $\go$ given the observed signal $k$. As such, a $K$-finite information design can be equivalently characterized by a set of measurable functions $\pi_k(\go),\ k=\{1,\cdots,K\}$ satisfying conditions $\pi_k(\go)\in[0,1]$ and $\sum_k\pi_k(\go)\ =\ 1$ with probability one.

Intuitively, an information design is a map from the space $\gO$ of possible states to a ``dictionary'' of $K$ messages, whereby the sender commits to a precise rule of selecting a signal from the dictionary for every realization of $\go.$ In principle, it is possible that this rule involves randomization, whereby, for a given $\go,$ the sender randomly picks a signal from a non-singleton subset of messages in the dictionary.
An information design does not involve randomization if and only if it is a partition of the state space $\gO.$ 

\begin{definition}[Randomization]\label{dfn-main}
We say that information design involves randomization if $\cP(\pi_k(\go)\not\in\{0,1\})>0$  for some $k$. We say that information design is a partition if $\cP(\pi_k(\go)\in\{0,1\})=1$ for all $k=1,\cdots,K.$ In this case,
\begin{equation}\label{part-1}
\cup_{k=1}^K \{\go:\ \pi_k(\go)=1\}
\end{equation}
is a Lebesgue-almost sure partition of $\gO$ in the sense that $\gO\setminus \cup_{k=1}^K \{\go:\ \pi_k(\go)=1\}$ has Lebesgue measure zero, and the subsets of the partition \eqref{part-1} are Lebesgue-almost surely disjoint.
\end{definition}

We use $\bar\pi\ =\ (\pi_k(\go))_{k=1}^K\in [0,1]^K$ to denote the random $K$-dimensional vector representing the information design. As we show below, a key implication of this setting is that, with a continuous state space and under appropriate regularity conditions, randomization is never optimal, and hence optimal information design is always given by a {\it partition}. While this result might seem intuitive, its proof is non-trivial and is based on novel techniques that, to the best of our knowledge, have never been used in the literature before. It is this result that is key our subsequent analysis of the unconstrained problem. 

\subsection{Receivers}

At time $t=0,$ upon observing a signal $k,$ each agent (receiver) $n=1,\cdots,N$ selects an action $a_n\in\R^m$ to maximize the expected utility function
\[
E[U_n(a_{n},a_{-n},\go)|k]\,,
\]
where we use $a_{-n}\ =\ (a_i)_{i\not=n}\in \R^{(N-1)m}$ to denote the vector of actions of other agents. We denote by  $a=(a_n)_{n=1}^N\in \R^M,\ M=Nm$ the vector of actions of all agents. A Nash equilibrium action profile $a(k)\ =\ (a_n(k))_{n=1}^N$ is a solution to the fixed point system
\begin{equation}\label{sys-1}
a_n(k)\ = \ \arg\max_{a_n} E[U_n(a_{n},a_{-n}(k),\go)|k]\,.
\end{equation}
We use $C^2(\gO)$ to denote the set of functions that are twice continuously differentiable in $\gO$. We will also use $D_a$ and $D_{aa}$ to denote the gradient and the Hessian with respect to the variable $a.$ Let 
\begin{equation}
G_n(a,\go)\ \equiv\ D_{a_n}U_n(a_{n},a_{-n},\go)\,.
\end{equation}

\begin{assumption}\label{ac} There exists an integrable majorant $Y(\go)\ge 0$ such that $Y(\go)\ge U_n(a_n,a_{-n},\go)$ for all $a\in \R^{M},\ \go\in\gO,\ n=1,\cdots,N.$ 
The function $U_n(a_n,a_{-n},\go)\in C^2(\R^m\times\R^{Nm}\times\gO)$ is strictly concave in $a_n$, and is such that $\lim_{\|a_n\|\to\infty}U_n=-\infty$.\footnote{This is a form of Inada condition ensuring that the optimum is always in the interior. An integrable majorant is needed to apply Fatou lemma and conclude that $\lim_{\|a_n\|\to\infty}E[U_n(a_n,a,\go)|k]= -\infty$ always.}

Furthermore, the map $G\ =\ (G_n)_{n=1}^N\,:\ \R^{M}\times \gO\to\R^{M}$ satisfies the following conditions:  
\begin{itemize}
\item $G$ is uniformly monotone in $a$ for each $\go$ so that $\eps\|z\|^2\ \le -z^\top D_aG(a,\go)z\le\eps^{-1}\|z\|^2$ for some $\eps>0$ and all $z\in \R^M;$\footnote{Strict monotonicity is important here. Without it, there could be multiple equilibria.} 

\item the unique solution $a_*(\go)$ to $G(a_*(\go),\go)=0$ is square integrable: $E[\|a_*(\go)\|^2]<\infty.$
\end{itemize}
\end{assumption}

Assumption \ref{ac} implies that the following is true:

\begin{lemma} \label{existence} For any information design $\pi,$ there exists a unique equilibrium $a\ =\ a_*(\pi)$. It is the unique solution to the fixed point system
\begin{equation}\label{sys-3}
E[G(a(k),\go)|k]\ =\ 0\,,
\end{equation}
and $\|a(k)\|^2\ \le\ \kappa E[\|a_*(\go)\|^2|k]$ for some universal $\kappa>0.$ This equilibrium depends smoothly on $\pi.$ 
\end{lemma}

\subsection{Optimal Information Design}

Without loss of generality, we may assume that at the optimum we always have $a(k)\not=a(\tilde k)$ for any $k\not=\tilde k.$ That is, different signals always induce different actions. We assume that the sender chooses the information design to maximize the expected public welfare function $W(a,\go)$ over all possible action profiles satisfying the participation (optimality) constraints of the receivers:
\begin{equation}\label{w-w}
\begin{aligned}
&\bar\pi^*\ =\ \arg\max_{\bar\pi}E[W(a_*(\bar\pi),\go)]\ =\ \arg\max_{\bar\pi}\{E[W(a,\go)]:\ a\ =\ {\rm maximizes\ agents'\ utilities}\}\\
&=\ \max_{\bar\pi,\ a}\{E[W(a(k),\go)]\ :\ E[G(a(k),\go)|k]\ =\ 0\ \forall\ k\}\,.
\end{aligned}
\end{equation}
By direct calculation, we can rewrite the expected social welfare function as
\begin{equation}\label{welfare-e}
E[W(a_*(\bar\pi),\go)]\ =\ \sum_{k=1}^K \int_\gO W(a_*(k,\bar\pi),\go)\,\pi_k(\go)\mu_0(\go)d\go\,.
\end{equation}

\begin{example}[Moment Persuasion]\label{exampl1} The most important example throughout this paper will be a setup where
\[
G_n(a,\go)\ =\ g_n(\go)\ -\ a_n\,
\]
for some functions $g_n(\go),\ n=1,\cdots,M.$ \cite{DworzakKolotilin2019} refer to this setup as ``moment persuasion." It is known that any continuous function $W(a,\go)$ can be uniformly approximated by a separable function, 
\[
W(a,\go)\ \approx\ \sum_{k=1}^\kappa f_k(a) \varphi_k(\go)
\]
for some smooth functions $\varphi_k,\ f_k$ (e.g., polynomials). As a result, defining $G_{n+l}=\varphi_l(\go)-a_{n+l},\ l=1,\cdots,\kappa,$ and 
\[
\tilde W(a)\ =\ \sum_{k=1}^\kappa f_k(a) a_{n+k}\,,
\]
we get 
\[
E[W(a,\go)]\ \approx\ E[\tilde W(a)]\,. 
\]
More generally, if we approximate 
\[
G(a,\go)\ \approx\ \sum_i \psi_i(a)\phi_i(\go)\,,
\]
we get that $E[G(a,\go)|k]=0$ is equivalent to the optimal action, $a,$ being a function of conditional expectations of $\phi_i.$ Thus, any optimal information design problem considered in this paper can be approximated by a moment persuasion problem. 
\end{example}

We also need a technical condition motivated by Example \ref{exampl1}.

\begin{assumption}\label{integrability} There exists a function $g:\gO\to \R_+$ such that $g(\go)\ge \|a_*(\go)\|^2$ and the set $\{\go:g(\go)\le A\}$ is compact for all $A>0,$ and a convex, increasing function $f\ge 1$ such that  $|W(a,\go)|+\|D_aW(a,\go)\|\ \le\ g(\go) f(\|a\|^2)$ and 
\begin{equation}\label{integrability1}
E[g^2(\go) f(g(\go))]<\infty.
\end{equation}
\end{assumption}

To state the main result of this section ---the optimality of partitions---we need also the following definition.

\begin{definition}
We say that functions $\{f_1(\go),\cdots,f_{L_1}(\go)\},\ \go\in\gO,$ are linearly independent modulo $\{g_1(\go),\cdots,g_{L_2}(\go)\}$ if there exist no real vectors $h\in \R^{L_1},\ k\in \R^{L_2}$ with $\|h\|\not=0,$ such that
\[
\sum_i h_i f_i(\go)\ =\ \sum_j k_j g_j(\go)\qquad for\ all\ \go\in\go\,.
\]
In particular, if $L_1=1,$ then $f_1(\go)$ is linearly independent modulo $\{g_1(\go),\cdots,g_{L_2}(\go)\}$ if $f_1(\go)$ cannot be expressed as a linear combination of $\{g_1(\go),\cdots,g_{L_2}(\go)\}.$
\end{definition}

We also need the following technical condition.

\begin{definition}\label{main-ass-indep} We say that $W, G$ are in a generic position if for any fixed $a,\tilde a\in \cR^N,\ a\not=\tilde a$, the function $W(a,\go)-W(\tilde a,\go)$ is linearly independent modulo $\left\{\{G_n(a,\go)\}_{n=1}^N,\{G_n(\tilde a,\go)\}_{n=1}^N\right\}$;
\end{definition}

$W, G$ are in generic position for generic functions $W$ and $G$.\footnote{The set of $W,G$ that are not in generic position is nowhere dense in the space of continuous functions.} We will also need a key property of real analytic functions\footnote{A function is real analytic if it can be represented by a convergent power series in the neighborhood of any point in its domain.} that we use in our analysis (see, e.g., \cite*{HugonnierMalamudTrubowitz2012}).

\begin{proposition}\label{zero-go} If a real analytic function $f(\go)$ is zero on a set of positive Lebesgue measure, then $f$ is identically zero. Hence, if real analytic functions $\{f_1(\go),\cdots,f_{L_1}(\go)\}$ are linearly dependent modulo $\{g_1(\go),\cdots,g_{L_2}(\go)\}$ on some subset $I\subset\gO$ of positive Lebesgue measure, then this linear dependence also holds on the whole $\gO$ except, possibly, a set of Lebesgue measure zero.
\end{proposition}

Using Proposition \ref{zero-go}, it is possible to prove the main result of this section:

\begin{theorem}[Optimal finite information design]\label{mainth1} There always exists an optimal $K$-finite information design $\bar\pi^*$ which is a partition. Furthermore, if $W,\ G$ are real analytic in $\go$ for each $a$ and are in generic position, then any $K$-finite optimal information design is a partition.
\end{theorem}

While Theorem \ref{mainth1} may seem intuitive, its proof is non-trivial and is based on novel techniques (see the appendix). First, we prove the 
second part. The existence of an optimum $\bar\pi$ follows by standard compactness arguments. Suppose now on the contrary that $\bar\pi$ is not a partition. Then, for some $k,$ we have $\pi_k(\go)\in (0,1)$ on some positive measure subset $I\subset \gO.$  At the global maximum, under arbitrary small perturbations, social welfare should decrease. We show that this can only be true if conditions of Definition \ref{main-ass-indep} are violated for $\go\in I.$ However, since $I$ has a positive Lebesgue measure, Proposition \ref{zero-go} implies that it has to be violated on the whole of $\gO$. Finally, the first part follows by a simple approximation argument because any function can be approximated by an analytic function satisfying the generic conditions of  Definition \ref{main-ass-indep}. Note that the regularity of both $W$ and the equilibrium (ensured by Lemma \ref{existence}) are crucial for the partition result. Without such regularity, classic examples of Bayesian persuasion (see, e.g., \cite{KamGenz2011}) show that randomization can be optimal. 

\subsection{The Structure of Optimal Partitions}

The goal of this section is to provide a general characterization of the ``optimal clusters" in Theorem \ref{mainth1}.

We use $D_aG(a,\go)\in \R^{M\times M}$ to denote the Jacobian of the map $G$, and, similarly, $D_aW(a,\go)\in \R^{1\times M}$ the gradient of the welfare function $W(a,\go)$ with respect to $a.$ For any vectors $x_k\ \in\ \R^{M},\ k=1,\cdots,K$ and actions $\{a(k)\}_{k=1}^K,$ let us define the partition
\begin{equation}\label{partitions1}
\begin{aligned}
\gO_k^*(\{x_\ell\}_{\ell=1}^K,\{a_\ell\}_{\ell=1}^K)\ &=\
\Bigg\{
\go\ \in\ \gO\ :\ W(a(k),\ \go)-x_k^\top G(a(k),\ \go)\\
& =\ \max_{1\le l\le K} \left(W(a(l),\ \go)\ -\ x_l^\top G(a(l),\ \go)\right)
\Bigg\}
\end{aligned}
\end{equation}
Equation \eqref{partitions1} is basically the first-order condition for the optimization problem \eqref{w-w}, whereby $x_k$ are the Lagrange multipliers of agents' participation constrains.

\begin{theorem}\label{regular-partition} Any optimal partition in Theorem \ref{mainth1} satisfies the following conditions:
\begin{itemize}
\item local optimality holds: $\gO_k\ =\ \gO_k^*(\{x_\ell\}_{\ell=1}^K,\{a_\ell\}_{\ell=1}^K)$ with $x_k^\top\ =\ \bar D_aW(k)(\bar D_aG(k))^{-1}\,,$
where we have defined for each $k=1,\cdots,K\,$
\begin{equation}
\begin{aligned}
&\bar D_aW(k)\ =\ \int_{\gO_k} D_aW(a(k),\go)\mu_0(\go)d\go\,,\ \bar D_aG(k)\ =\ \int_{\gO_k} D_{a}G(a(k),\go)\mu_0(\go)d\go
\end{aligned}
\end{equation}
\item the actions $\{a(k)\}_{k=1}^K$ satisfy the fixed point system
\begin{equation}\label{gak1}
\int_{\gO_k} G(a(k),\go)\mu_0(\go)d\go\ =\ 0,\ k=1,\,\cdots,\,K\,.
\end{equation}
\item the boundaries of $\gO_k$ are a subset of the variety\footnote{This variety is real analytic when so are $W$ and $G.$ A real analytic variety in $\R^m$ is a subset of $\R^m$ defined by a set of identities $f_i(\go)=0,\ i=1,\cdots,I$ where all functions $f_i$ are real analytic. If at least one of functions $f_i(\go)$ is non-zero, then a real analytic variety is always a union of smooth manifolds and hence has a Lebesgue measure of zero. When $W,G$ are real analytic and are in generic position, the variety $\left\{\go\in \R^m\ :\ W(a(k),\ \go)\ -\ x_k^\top G(a(k),\ \go)\ =\ W(a(l),\ \go)\ -\  x_l^\top G(a(l),\ \go)\right\}$ has a Lebesgue measure of zero for each $k\not=l.$}
\begin{equation}\label{indiff}
\cup_{k\not=l}\left\{\go\in \R^m:W(a(k),\go)-x_k^\top G(a(k),\go)=W(a(l),\go)-x_l^\top G(a(l),\go)\right\}\,.
\end{equation}
\end{itemize}
\end{theorem}

A key insight of Theorem \ref{regular-partition} comes from the characterization of the different clusters of an optimal partition. The sender has to solve the problem of maximizing social welfare \eqref{welfare-e} by inducing the desired actions, $(a_n),$ of economic agents for every realization of $\go.$ Ideally, the sender would like to induce $a_*\ =\ \arg\max_a W(a,\go)\,.$ However, the ability of the sender to elicit the desired actions is limited by the participation constraints of the receivers – that is, the map from the posterior beliefs induced by communication to the actions of the receivers. Indeed, while the sender can induce any Bayes-rational beliefs (i.e., any posteriors consistent with the Bayes rule; see \cite{KamGenz2011}), she has no direct control over the actions of the receivers. The degree to which these constraints are binding is precisely captured by the Lagrange multipliers $x(a),$ so that the sender is maximizing the Lagrangian $\max_a(W(a,\go)-x(a)^\top G(a,\go)).$  Formula \eqref{partitions1} shows that, inside the cluster number $k$, the optimal action profile maximizes the respective Lagrangian. The boundaries of the clusters are then determined by the indifference conditions \eqref{indiff}, ensuring that at the boundary between regions $k$ and $l$ the sender is indifferent between the respective action profiles $a_k$ and $a_l.$

Several papers study the one-dimensional case (i.e., when $L=1$ so that $\go\in \R^1$) and derive conditions under which the optimal signal structure is a monotone partition into intervals. Such a monotonicity result is intuitive, as one would expect the optimal information design to only pool nearby states. The most general results currently available are due to \cite{Hopenhayn2019} and \cite{DworczakMartini2019},\footnote{See also \cite{mensch2018}.} but they cover the case when sender's utility (social welfare function in our setting) only depends on $E[\go]\in \R^1.$\footnote{This is equivalent to $G(a,\go)\ =\ a-\go$ in our setting. In this case, formula \eqref{gak1} implies that the optimal action is given by $a(k)=E[\go|k].$} Under this assumption, \cite{DworczakMartini2019} derive necessary and sufficient conditions guaranteeing that the optimal signal structure is a monotone partition of $\gO$ into a union of disjoint intervals. \cite{Arielietal2020} (see, also, \cite{kleiner2020extreme}) provide a full solution to the information design problem when $a(k)=E[\go|k]$ and, in particular, show that the partition result does not hold in general when the signal space is continuous. 
Theorem \ref{regular-partition} proves that a $K$-finite optimal information design is in fact always a partition when the state space is continuous and the signal space is discrete. However, no general results about the monotonicity of this partition can be established without imposing more structure on the problem.\footnote{Of course, as \cite{DworczakMartini2019} and \cite{Arielietal2020} explain, even in the one-dimensional case the monotonicity cannot be ensured without additional technical conditions. No such conditions are known in the multi-dimensional case. \cite{DworczakMartini2019} present an example with four possible actions $(K=4)$ and a two-dimensional state space $(L=2)$ for which they are able to show that the optimal information design is a partition into four convex polygons. See, also, \cite{DworzakKolotilin2019}.}

Consider the optimal information design of Theorem \ref{regular-partition} and define the piece-wise constant function\,. 
\begin{equation}
a(\go)\ =\ \sum_k a(k) {\bf 1}_{\go \in \gO_k}
\end{equation}
Since welfare can be written as 
\[
E[W(a(\go),\go)],
\]
this function encodes all the properties of the optimal information design. It turns out that in the case of moment persuasion (Example \ref{exampl1}), optimality conditions of Theorem \ref{regular-partition} can be used to derive important, universal monotonicity properties of the optimal policy function $a(\go).$ 

\begin{proposition}\label{main convexity} Suppose that we are in a moment persuasion setup: $G=g(\go)-a$ with $g(\go):\ \gO\to \R^M$ and $W(a,\go)=W(a).$ 
Suppose also that $M\le L.$ Let $X\subset \gO$ be an open set suppose that $g$ is injective on $X$ and $g(X)$ is convex. Then, the set $g(\gO_k\cap X)$ is convex for each $k.$ In particular, $\gO_k\cap X$ is connected. Furthermore, the map $x\to D_aW(a(g^{-1}(x)))$ is monotone increasing on $g(\gO_k\cap X).$\footnote{A map $F$ is monotone increasing if $(x_1-x_2)^\top (F(x_1)-F(x_2))\ge 0$ for all $x_1,x_2.$} 
\end{proposition}

Partitions into convex sets are the natural multi-dimensional analogs of one-dimensional monotone partitions from \cite{Hopenhayn2019} and \cite{DworczakMartini2019}.
Injectivity of the $g$ map and convexity of its image are crucial for the connectedness of the $\gO_k$ regions. Without injectivity, even bounding the number of connected components of $\gO_k$ is non-trivial. These effects become particularly strong in the limit when $K\to\infty,$ where lack of injectivity in the map $g$ may lead to a breakdown of even minimal regularity properties of the optimal map. 

\section{The Unconstrained Problem and the Cost of Information Transport}

The full, unconstrained optimal information design problem can be formulated as follows (see, e.g., \citet{KamGenz2011}, \cite{DworzakKolotilin2019}): 

\begin{definition}\label{def-u} Let $\Delta(\gO)$ be the set of probability measures on $\gO$ and define $\Delta(\Delta(\gO))$ similarly. Let also $a(\mu)$ be the unique solution to 
\[
\int_\gO G(a,\go)d\mu(\go)\ =\ 0
\]
and let 
\[
\bar W(\mu)\ =\ \int_\gO W(a(\mu),\go)d\mu(\go)\,.
\]
The optimal Bayesian persuasion (optimal information design) problem is to maximize 
\[
\int_{\gD(\gO)} \bar W(\mu)d\tau(\mu)
\]
over all distributions of posterior beliefs $\tau\in\Delta(\Delta(\gO))$ satisfying 
\[
\int_{\gD(\gO)} \mu d\tau(\mu)\ =\ \mu_0\,. 
\]
Such a $\tau$ is called an information design. We say that $\tau$ does not involve randomization if there exists a map $a:\gO\to \R^M$ such that $\tau$ coincides with the set of distributions of $\go$ conditional on $a.$ In this case, $a(\go)$ will be referred to as an optimal policy. 
\end{definition}

We start with the following important technical lemma. The proof of this lemma is non-trivial due to additional complications created by the potential non-compactness of the set $\gO.$ 

\begin{lemma}\label{lem-approx} When $K\to\infty,$ maximal social welfare attained with $K$-finite optimal information designs converges to the maximal welfare attained in the full, unconstrained problem of Definition \ref{def-u}. 
\end{lemma}

The next result is the key step in deriving properties of an optimal information design. We will use $\Supp(a)$ to denote the support of any map $a:$ 
\[
\Supp(a)\ =\ \{x\in \R^M: \mu_0(\{\go:\ \|a(\go)-x\|<\eps\})>0\ \forall\ \eps>0\}\,. 
\]

\begin{definition}\label{full-def} Let $a_*(\go)$ be the unique solution to $G(a_*(\go),\go)=0.$ For any map $x:\R^M\to \R^M,$ we define 
\begin{equation}\label{c-function}
c(a,\go;x)\ \equiv\ W(a_*(\go),\go)-W(a,\go)\ +\ x(a)^\top G(a,\go)\,.
\end{equation}
For any set $\Xi\subset \R^M$, we define 
\begin{equation}
\phi_\Xi(\go;x)\ \equiv\ \inf_{a\in \Xi} c(a,\go;x)\,. 
\end{equation}
Everywhere in the sequel, we refer to $c$ as the {\bf cost of information transport.}
\end{definition}

To gain some intuition behind the cost $c,$ we note that $W(a_*(\go),\go)$ is the welfare attained by revealing that the true state is $\go.$ Thus, $W(a,\go)-W(a_*(\go),\go)$ is the welfare gain from inducing a different (preferred) action $a$ and $x(a)^\top G(a,\go)$ is the corresponding shadow cost of agents' participation constraints. The total cost of information transport is the sum of the true and the shadow costs of ``transporting" information from $a_*(\go)$ to $a.$ 

Since $E[W(a_*(\go),\go)]$ is independent of the information design, maximizing expected welfare $E[W(a(\go),\go)]$ is equivalent to minimizing $E[W(a_*(\go),\go)-W(a(\go),\go)].$ From now on, we will be considering this equivalent formulation of the problem. Note that, for any policy $a$ satisfying $E[G(a(\go),\go)|a(\go)=a]=0$ and any well-behaved $x$ we always have 
\begin{equation}\label{key}
E[W(a_*(\go),\go)-W(a(\go),\go)]\ =\ E[c(a(\go),\go;x)]\,. 
\end{equation}
Thus, the problem of maximizing $E[W(a(\go),\go)]$ over all admissible policies $a(\go)$ is equivalent to the problem of {\it minimizing the expected cost of information transport},  $E[c(a(\go),\go;x)]\,.$

By passing to the limit in Theorem \ref{regular-partition}, it is possible to prove the following result. 

\begin{theorem}\label{mainth-limit} There exists a Borel-measurable optimal Bayesian persuasion solving the problem of Definition \ref{def-u} that does not involve randomization. The corresponding optimal map $a(\go)$ satisfies $
E[G(a(\go),\go)|a(\go)=a]\ =\ 0$
for all $a\in \R^M.$ Furthermore, if we define 
$\Xi\ =\ \Supp(a)$
and 
\begin{equation}\label{x-top}
x(a)^\top\ =\ E[D_aW(a,\go)|a]\,E[D_aG(a,\go)|a]^{-1}\,,
\end{equation}
then 
\begin{itemize}
\item we have 
\begin{equation}\label{c<0}
c(a(\go),\go;x)\le 0
\end{equation}

\item we have 
\begin{equation}\label{olicy}
a(\go)\ =\ \arg\min_{b\in \Xi}c(b,\go;x)
\end{equation}
and the function $c(a(\go),\go;x)$ is Lipschitz continuous in $\go.$ 

\item we have $E[W(a_*(\go),\go)-W(a(\go),\go)]\ =\ E[\phi_\Xi(\go;x)]$
\end{itemize}
Furthermore, any optimal information design satisfies \eqref{c<0} and \eqref{olicy}.
\end{theorem}

Theorem \ref{mainth-limit} characterizes some important properties that are necessary for an optimal information design. It turns out that optimal policies possess certain remarkable properties that can be established using results from optimal transport theory. We first recall the classical optimal transport problem of Monge and Kantorovich (see, e.g., \cite{mccann2011five}). 

\begin{definition} Consider two probability measures, $\mu_0(\go)d\go$ on $\gO$ and $\nu$ on $\Xi.$ The optimal map problem (the Monge problem) is to find a map $X:\gO\to \Xi$ that minimizes 
\[
\int c(X(\go),\go)\mu_0(\go)d\go
\]
under the constraint that the random variable $\chi=X(\go)$ is distributed according to $\nu.$
The Kantorovich problem is to find a probability measure $\gamma$ on $\Xi\times \gO$ that minimizes 
\[
\int c(\chi,\go)d\gamma(\chi,\go)
\]
over all $\gamma$ whose marginals coincide with $\mu_0(\go)d\go$ and $\nu$, respectively. 
\end{definition}

It is known that, under very general conditions, the Monge problem and its Kantorovich relaxation have identical values, and an optimal map exists. 
It turns out that, remarkably, any optimal policy solves the Monge problem.\footnote{This result is a direct analog of Corollary 2.5 in \cite{kramkov2019optimal} that was established there in the special case of $W(a)=a_1a_2,$ $g(\go)=\binom{\go_1}{\go_2}$ and $L=M=2.$ Its proof is also completely analogous to that in \cite{kramkov2019optimal}.}

\begin{theorem}\label{monge} Any optimal policy $a(\go)$ solves the Monge problem with $\nu$ being the distribution of the random vector $a(\go).$ 
\end{theorem}

The result of Theorem \ref{monge} puts us in a perfect position to apply all the powerful machinery of optimal transport theory and derive properties of optimal policies. 

\begin{definition} We say that a map $a:\gO\to\R^M$ is $c$-cyclically monotone if and only if all $k\in \N$ and $\go_1,\cdots,\go_k\in \gO$ satisfy 
\[
\sum_{i=1}^k c(a(\go_i),\go_i;x)\ \le\ \sum_{i=1}^k c(a(\go_i),\go_{\sigma(i)};x)
\]
for any permutation $\sigma$ of $k$ letters.
\end{definition}

We start with the following result which is a direct consequence of a theorem of \cite{smith1992hoeffding}. 

\begin{corollary}\label{cor-cyclic} Any optimal policy is $c$-cyclically monotone.\footnote{This result is, in fact, a direct consequence of \eqref{olicy} because each $\go$ is already coupled with the ``best" $a(\go).$} 
\end{corollary}

Cyclical monotonicity plays an important role in the theory of optimal demand. See, for example, \cite{rochet1987necessary}. In our setting, this result has a similar flavour: In order to induce an optimal action, the sender optimally aligns actions $a$ with the state $\go$ to minimize the cost of information transport, $c.$   

We complete this section with a direct application of a theorem of \cite{gangbo1995habilitation} and \cite{levin1999abstract}. Recall that a function $u(\go)$ is called $c$-convex if and only if $u=(u^{\tilde c})^c$ where 
\[
u^{\tilde c}(a)\ =\ \sup_{\go\in \gO}(-c(a,\go)-u(\go)),\ v^c(\go)\ =\ \sup_{a\in \R^M}(-c(a,\go)-v(a))
\]
The class of $c$-convex functions has a lot of nice properties (see, e.g., \cite{mccann2011five}). In particular, a $c$-convex function $u$ is twice differentiable Lebesgue-almost everywhere and satisfies $|Du(\go)|\ \le\ \sup_a |D_\go c(a,\go)|$ and $D_{aa}^2u(\go)\ \ge\ \inf_a-D_{\go\go} c(,\go).$ The following is true (see, e.g., \cite{mccann2011five}).

\begin{corollary}\label{gradentcor} Suppose that $c$ is jointly continous in $(a,\go)$ and that the map $a\ \to\ D_\go c(a,\go)$ is injective for Lebesgue-almost every $\go\in\gO.$ Let $a=Y(\go, p)$ be the unique solution to $D_\go c(a,\go)=-p$ for any $p.$ Then, there exists a locally Lipschitz $c$-convex function $u:\gO\to \R$ such that $a(\go)\ =\ Y(\go, Du(\go)).$ 
\end{corollary}

Corollaries \ref{cor-cyclic} and \ref{gradentcor} provide strong necessary conditions for an optimal policy. 
Unfortunately, in general we do not know if the conditions of Theorem \ref{mainth-limit} are also sufficient for optimality. As we show in the next section, such sufficiency can be established in the setting of moment persuasion which, as we explain above (see Example \ref{exampl1}), can in fact be used to  approximate any optimal information design problem. 

\section{Moment Persuasion} 

In the case of moment persuasion, $G(a,\go)=a-g(\go),\ a_*(\go)=g(\go),\ W(a,\go)=W(a)$ and $x(a)^\top=D_aW(a).$ The key simplification comes from the fact that $x(a)$ is independent of the choice of information design. We will slightly abuse the notation and introduce a modified definition of the function $c.$ Namely, we define 
\begin{equation}\label{cab}
c(a,b)\ =\ W(b)\ -\ W(a)\ +\ D_aW(a)\,(a-b)
\end{equation}
As one can see from \eqref{cab}, the cost of information transport, $c,$ coincides with the classic Bregman divergence that plays an important role in convex analysis (see, e.g., \cite{rockafellar1970convex}). A key innovation in this paper is the introduction of Bregman divergence with a non-convex $W.$ In this case, none of the classic results about Bregman divergence hold true, and new techniques need to be developed. 

Our key objective here is to understand the structure of the support set $\Xi$ of an optimal policy. Here, the nature of the map $g:\ \gO\to R^M$ will play an important role, as can already be guessed from Proposition \ref{main convexity}. 

\begin{definition} Let $conv(X)$ be the closed convex hull of a set $X\subset\R^M.$ A set $\Xi\subset \R^M$ is $X$-maximal if $\inf_{a\in \Xi}c(a,b)\le 0$ for all $b\in X.$ A set $\Xi$ is $W$-monotone if $c(a_1,a_2)\ge 0$ for all $a_1,a_2\in \Xi.$ A set $\Xi$ is $W$-convex if $W(ta_1+(1-t)a_2)\le t W(a_1)+(1-t)W(a_2)$ for all $a_1,a_2\in\Xi,\ t\in [0,1].$ We also define 
\[
\phi_\Xi(x)\ =\ \inf_{a\in \Xi} c(a,x)\,.
\]
Note that we are again slightly abusing the notation so that the function $\phi_\Xi$ from the previous section corresponds to $\phi_{\Xi}(g(\go))$ in this section.
\end{definition}

\begin{lemma}\label{w-convex}
Every $W$-convex set is $W$-monotone. 
\end{lemma}

Indeed, the function $q(t)=t W(a_1)+(1-t)W(a_2)-W(ta_1+(1-t)a_2)$ satisfies $q(t)\ge q(0)=0$ and hence $0\le q'(0)\ =\ c(a_2,a_1).$ 
For the readers' convenience, we now state an analog of Theorem \ref{mainth-limit} for the case of moment persuasion. Recall also that (by \eqref{key}) we are solving the equivalent problem of minimizing $E[c(a(\go),g(\go))].$

\begin{corollary}\label{cor-moment} Suppose that $W=W(a),\ G=a-g(\go)$ are such that $|W(a)|+\|D_aW\|\le f(\|a\|^2)$ for some convex function $f$ satisfying $E[\|g(\go)\|^2f(\|g(\go)\|^2)]<\infty.$ Then, there exists an optimal policy $a(\go)$ with $\Xi=\Supp(a)$ such that:
\begin{itemize}
\item $a(\go)\ =\ E[g(\go)|a(\go)]$, $a(\go)\ =\ \arg\min_{b\in \Xi}c(b,g(\go))$, the function $c(a(\go),g(\go))$ is Lipschitz continuous in $\go.$ 

\item if $g$ is injective on a set $X\subset\gO$ and $g(X)$ is convex, then $D_aW(a(g^{-1}(x)))$ is monotone increasing on $X$, while $c(a(g^{-1}(x)),x)$ is convex on $X$ and $D_aW(a(g^{-1}(x)))$ is a subgradient of $c(a(g^{-1}(x)),x).$ 
\end{itemize}
Furthermore, and any optimal policy satisfies $c(a(\go),g(\go))\le 0\,.$ 
\end{corollary}

The main result of this section is the following theorem which provides explicit and verifiable necessary and sufficient conditions for optimality of a given policy. 

\begin{theorem}\label{converse} For any optimal policy $a(\go),$ the set $\Xi=\Supp(a)$ is $g(\gO)$-maximal and $W$-convex. Furthermore, any policy $a(\go)$ satisfying $a(\go)\ =\ E[g(\go)|a(\go)]$, and $a(\go)\ =\ \arg\min_{b\in \Xi}c(b,g(\go))$, and such that $\Xi=\Supp(a)$ is $conv(g(\gO))$-maximal\footnote{Surprisingly, its $W$-convexity (and, hence, $W$-monotonicity) then follows automatically.} is optimal. 
\end{theorem}

Intuitively, it is optimal to reveal information along the ``domains of convexity" of $W$. Hence the $W$-convexity of the support $\Xi$ of an optimal policy. 
$W$-convexity of $\Xi$ also implies $W$-monotonicity. It means that transporting information along $\Xi$ is costly, as information on $\Xi$ is already in its optimal ``location". Similarly, maximality of $\Xi$ means that any point outside of $\Xi$ can be transported to some location on $\Xi$ at a negative cost, improving the overall welfare. The arguments in the proof of Theorem \ref{converse} can also be used to shed some light on the uniqueness of optimal policies. 

\begin{proposition}\label{uniqueness} Let $a(\go)$ be an optimal policy with support $\Xi.$ Let also $Q_\Xi$ be the set $\{b\in \R^M:\ \phi_\Xi(b)=0\}.$ Then, if $\Xi$ is $X$-maximal for some set $X$ and $\Xi\subseteq X,$ we have $\Xi\subseteq Q_\Xi.$ Furthermore,  if $\tilde a$ is another optimal policy with support $\tilde\Xi,$ then   $\tilde\Xi\subseteq Q_{\Xi}$.

If $\Xi=Q_\Xi$ and $\arg\min_{a\in \Xi}c(a,b)$ is a singleton for all $b\in conv(g(\gO)),$ then the optimal policy is unique. 
\end{proposition}

We conjecture that the conditions in Proposition \ref{uniqueness} hold generically and hence optimal policy is unique for generic $W.$ This is indeed the case in the explicitly solvable examples discussed in Section \ref{section-examples}. 

We now discuss other, more subtle properties of optimal policies. We start with an application of Corollary \ref{gradentcor}. 

\begin{corollary}\label{inject} Suppose that $M\le L$ and $D_\go g$ has rank $M$ for Lebesgue-almost every $\go$ and that $a\to D_aW(a)$ is injective. Then, the map $a\to D_\go c(a,g(\go))$ is injective. Let $Y(\go, p)$ be the unique solution to $D_\go c(a,g(\go))=-p$. 
Then there exists a locally Lipschitz $c$-convex function $u(\go)$ such that $a(\go)=Y(\go, Du(\go)).$ In particular, $a(\go)$ is differentiable Lebesque-almost everywhere. 
\end{corollary}

Our next objective to get an idea about the ``amount" of information revealed by an optimal policy. When $W(a)$ is convex (so that $D_{aa}W(a)$ is positive semi-definite for any $a$), then full revelation is an optimal policy, and it is the only optimal policy if $W$ is strictly convex. Thus, in this case, $\Xi=g(\gO)$ may have dimension up to $M$. By contrast, if $W$ is strictly concave (so that $D_{aa}W(a)$ is negative semi-definite), then revealing no information is optimal and $\Xi=\{E[g(\go)]\}$ has a dimension of zero. But what happens when $W$ is neither concave nor convex? In this case, it is natural to expect that $\Xi$ will be ``smaller" than $\R^M$ and its ``smallness" depends on a ``degree of concavity" of $W$. One natural candidate for such a degree is the number of negative eigenvalues of $D_{aa}W.$ And indeed, as \cite{tamura2018bayesian} shows, when (1) $W(a)=a^\top H a+b^\top a$ for some matrix $H\in \R^{L\times L}$ and some vector $b\in \R^L$; (2) $g(\go)=\go$ and (3) $\mu_0(\go)$ is a multi-variate Gaussian distribution, there exists an optimal policy $a(\go)$ with $\Xi$ being a subspace of $\R^M$ whose dimension equals the number of nonnegative eigenvalues of $H.$ A key simplifying property of the linear optimal policy in \cite{tamura2018bayesian} is its regularity: A linear subspace is a smooth manifold and, hence, has a natural notion of dimension. However, even in the simple setup of \cite{tamura2018bayesian}, nothing is known about the behaviour of other policies. Are there non-linear policies? If yes, how much information to they reveal? 
Here, we establish a surprising result linking the number of positive eigenvalues of the Hessian $D_{aa}W$ with the Hausdorff dimension of the support $\Xi$ of {\it any} optimal policy. Recall that the d-dimensional Hausdorff measure $\cH^d$ of a subset $S\subset \R^M$ is defined as 
\[
\cH^d(S)\ =\  \lim_{r\to 0}\inf\{\sum_i r_i^d:\ \text{there is a cover of $S$ by balls with radii}\ 0<r_i<r\}\,,
\]
and the Hausdorff dimension $\dim_H(S)$ is defined as 
\[
\dim_H(S)\ =\ \inf\{d\ge 0:\ \cH^d(S)\ =\ 0\}\,. 
\]
It is known that Hausdorff dimension coincides with the ``natural" definition of dimension for sufficiently regular sets. E.g., $\dim_H(S)=d$ for a smooth, $d$-dimensional manifold $S$. However, in general, we do not have any strong regularity results for the behaviour of the support $\Xi$ of an optimal policy $a(\go)$. Hence, proving that $\Xi$ is a smooth manifold seems in general out of reach and $\Xi$ may potentially be highly irregular. For irregular sets, Hausdorff dimension may behave in a very complex fashion and may even take fractional values for fractals. 
The following is true. 

\begin{theorem}\label{frostman} Let $X$ be a Borel set, $X\subset \R^M.$ Suppose that either $D_{aa}W$ is a constant matrix or that $D_{aa}W(a)$ is continuous in $a$ and is non-degenerate for all $a\in X$ except for a countable set of points.\footnote{Alternatively, one can impose bounds on the Hausdorff dimension of the set $\{a:\det(D_{aa}W(a))=0\}.$} Let $\nu(a)$ be the number of nonnegative eigenvalues of $D_{aa}W(a).$ Then, for {\bf any} optimal policy $a(\go),$ we have 
\[
\dim_H(\Supp(a)\cap X)\ \le\ \sup_{a\in X}\nu(a)\,.
\]
\end{theorem} 

\section{Examples} \label{section-examples}

In this section, we investigate several concrete examples illustrating applications of Theorem \ref{converse}. All our examples will be based on the following technical result which is a direct consequence of Theorem \ref{converse}. 

\begin{proposition} Let $F$ be a  bijective, bi-Lipshitz map,\footnote{A map $F$ is bi-Lipschitz if both $F$ and $F^{-1}$ are Lipschitz continuous.} $F:\ X\to \gO$ for some open set $X\subset \R^M.$ Let also $M_1\le M$ and $x=(x_1,x_2)$ with $x_1\in X_1$, the projection of $X$ onto $\R^{M_1}$ and $x_2\in X_2,$ the projection of $X$ onto $\R^{L-M_1}.$ Define 
\begin{equation}\label{cond-f}
\begin{aligned}
&f(x_1)\ =\ f(x_1;F)\ \equiv\ \frac{\int_{X_{2}} |\det(D_xF(x_1,x_2))|\mu_0(F(x_1,x_2))\,g(F(x_1,x_2))dx_2}
{\int_{X_{2}} |\det(D_xF(x_1,x_2))|\mu_0(F(x_1,x_2))dx_2}\,.
\end{aligned}
\end{equation}
Suppose that $f$ is an injective map, $f:X_1\to \R^{M}$ and define 
\begin{equation}\label{opt-dif}
\phi(b)\ =\ \min_{y_1\in X_1}\{W(b)-W(f(y_1))+D_aW(f(y_1))^\top(f(y_1)-b)\}\,.
\end{equation}
Suppose also that the min in \eqref{opt-dif} for $b=g(F(x_1,x_2))$ is attained at $y_1=x_1$ and that
$\phi(b)\ \le\ 0\ \forall\ b\in conv(g(\gO)).$ Then, $a(\go)=f((F^{-1}(\go))_1)$ is an optimal policy. If $x_1=\arg\min$ in \eqref{opt-dif} with $b=F(x_1,x_2)$ for all $x_1,x_2$ and $\phi(b)<0$ for all $b\in conv(g(\gO))\setminus \Xi$ with $\Xi=f(X_1),$ then the optimal policy is unique. 

If $f$ is Lipshitz-continuous and the minimum in \eqref{opt-dif} is attained at an interior point, we get a system of second order partial integro-differential equations for the $F$ map: 
\begin{equation}\label{pde}
D_{x_1}f(x_1)^\top D_{aa}W(f(x_1))(f(x_1)-g(F(x_1,x_2)))\ =\ 0\,.
\end{equation}
\end{proposition}

We start with the simplest setting: A quadratic problem with $W(a)=a^\top H a.$ In this case, Theorem \ref{converse} implies that, for any optimal policy, $\Xi$ has to be monotonic, meaning that $(a_1-a_2)^\top H(a_1-a_2)\ge 0$ for all $a_1,a_2\in \Xi.$ The question we ask is: Under what conditions is $a(\go)=A\go$ with some matrix $A$ of rank $M_1\le M$ is optimal with $g(\go)=\go.$ Clearly, it is necessary that $\mu_0$ have linear conditional expectations,\footnote{This is, e.g., the case for all elliptical distributions, but also for many other distributions. See \cite{wei1999linear}.} $E[\go|A\go]\ =\ A\go$.  But then, since $E[A\go|A\go]=A\go,$ we must have $A^2=A,$ so that $A$ is necessarily a projection. Maximal monotonicity implies that $Q=A^\top H A$ is positive semi-definite, and\footnote{Here, $Q^{-1}$ is the Moore-Penrose inverse.} 
\[
\phi(b)\ =\ \min_{\go}(b^\top Hb+\go^\top A^\top H (A \go-2b))\ =\ b^\top HAQ^{-1}A^\top A(AQ^{-1}A^\top H -2Id)b
\]
with the minimizer $Q^{-1}A^\top H b.$ Thus, $A$ satisfies the fixed point equation $A\ =\ Q^{-1}A^\top H$ and hence $A^\top=HAQ^{-1}.$ Furthermore, maximality of $\Xi$ implies that 
\[
H\ +\ HAQ^{-1}A^\top A(AQ^{-1}A^\top H -2Id)\ =\ H-A^\top A
\]
is negative semi-definite. As a result, $(Id-A^\top)(H-A^\top A)(Id-A)=(Id-A^\top)H(Id-A)$ is also negative semi-definite, implying that $A$ and $Id-A$ ``perfectly split" positive and negative eigenvalues of $H$. Here, it is instructive to make two observations: First, optimality requires that $a(\go)$ ``lives" on positive eigenvalues of $H$. Second, maximality (the fact that $\phi(b)\le 0$ for all $b$) requires that $A$ absorbs all positive eigenvalues, justifying the term ``maximal". By direct calculation, we obtain the following extension of the result of \cite{tamura2018bayesian} (uniqueness follows from Proposition \ref{uniqueness}). 

\begin{corollary}\label{tamura} Suppose that $W(a)=a^\top H a$. Suppose that $\go$ has an elliptical distribution with a density $\mu_0(\go)=\mu_*(\go^\top \Sigma^{-1}\go)$ for some $\mu_*.$ Let $V=\Sigma^{1/2}H\Sigma^{1/2}$. Define $Q_+=[q_1,\cdots,q_r]$ as a $k\times r$ matrix consisting of the eigenvectors $q_1,\cdots,q_r$ associated with all the positive eigenvalues of $V$.  Then,  
\[
a(\go)\ =\ \Sigma^{1/2}Q_+(Q_+'Q_+)^{-1}Q_+'\Sigma^{-1/2}\go\,
\]
is an optimal policy. Furthermore, if $\det(H)\not=0,$ then the optimal policy is unique. In particular, there are no non-linear optimal policies. 
\end{corollary}

%
%
Consider now a non-linear version of this problem. Suppose that $W(a)\ =\ \varphi(\|a\|^2)$ for some smooth function $\varphi$. Let $F$ correspond to multi-dimensional spherical coordinates: $\go=F(r,\theta)=r x(\theta)$ where $\theta$ are the angular coordinates on the unit sphere and $r=\|\go\|.$ Then, with $x_1=\theta,x_2=r$ we have 
\[
\frac{\int_{X_{2}} |\det(D_xF(x_1,x_2))|\mu_0(F(x_1,x_2))\,g(F(x_1,x_2))dx_2}{\int_{X_{2}} |\det(D_xF(x_1,x_2))|\mu_0(F(x_1,x_2))\,dx_2}\ =\ \frac{\int_{0}^\infty r^{L-1} \mu_0(rx(\theta))g(rx(\theta))dr}{\int_{0}^\infty r^{L-1} \mu_0(rx(\theta))dr}\,.
\]
Suppose that $g$ and $\mu_0$ are such that 
\[
\frac{\int_{0}^\infty r^{L-1} \mu_0(rx(\theta))g(rx(\theta))dr}{\int_{0}^\infty r^{L-1} \mu_0(rx(\theta))dr}\ =\ \ga x(\theta)\,.
\]
For example, this is the case when $g(\go)=\go\,\psi(\|\go\|)$ for some function $\psi\ge 0$ and $\mu_0(\go)=\mu_*(\|\go\|^2)$ is spherically symmetric (a special case of an elliptical distribution). Thus, we must have $a(\go)=\ga \go/\|\go\|$. In this case, $\Xi$ is the sphere $\|a\|=\ga$ and we have 
\begin{equation}
\phi(b)\ =\ \min_{a\in \Xi}(\varphi(\|b\|^2)-\varphi(\ga^2)+2\varphi'(\ga^2)a^\top (a-b))\ =\ \varphi(\|b\|^2)-\varphi(\ga^2)+2\varphi'(\ga^2)\ga (\ga-\|b\|)\,,
\end{equation}
and the minimizer is $a=\ga b/\|b\|.$ Thus, we get the fixed point equation $\ga \go/\|\go\|=a(\go)=\ga g(\go)/\|g(\go)|=\ga \go/\|\go\|.$ Maximality is achieved when $\phi(b)$ is always non-positive, for all $\|b\|\le \max_{x\le R} (x\psi(x)).$ Thus, we arrive at the following result. 

\begin{corollary} Suppose that $g(\go)=\go\,\psi(\|\go\|^2)$ for some function $\psi\ge 0$ and $\mu_0(\go)=\mu_*(\|\go\|^2)$ and $W(a)\ =\ \varphi(\|a\|^2)$ and let 
\[
\ga\ \equiv\ \frac{\int_{0}^\infty r^{L} \mu_*(r^2)\psi(r^2)dr}{\int_{0}^\infty r^{L-1} \mu_*(r^2)dr}\,.
\]
If 
\begin{equation}\label{dwor1}
\max_{\|b\|\le \max_{x\ge 0} (x\psi(x^2))}(\varphi(\|b\|^2)-\varphi(\ga^2)+2\varphi'(\ga^2)\ga (\ga-\|b\|))\ \le\ 0\,, 
\end{equation}
then $a(\go)=\ga \go/\|\go\|$ is an optimal policy and the optimal policy is unique if the maximum in \eqref{dwor1} is attained only when $\ga=\|b\|.$ 
\end{corollary}

Consider now a more complex example where $\R^L=\R^{L_1}\oplus \R^{L_2}$ and $W(a)=\varphi(\|a_1\|^2,a_2)$ where $\varphi(y_1,y_2)$ is monotone increasing in $y_1$ and satisfies $\varphi(y_1,0)\ge \varphi(y_1,y_2)$ for all $y_2.$ Let $\Xi\subset\binom{\R^{L_1}}{0}$ be the sphere $\|\go_1\|=\ga$ whose Hausdorff dimension is $L_1-1$ (see Theorem \ref{frostman}). Then, 
\begin{equation}
\begin{aligned}
&\phi(b)\ =\ \min_{a\in \Xi}(\varphi(\|b_1\|^2,b_2)-\varphi(\ga^2,0)+2\varphi_{y_1}(\ga^2,0)a^\top (a-b))\\ 
&=\ \varphi(\|b_1\|^2,b_2)-\varphi(\ga^2,0)+2\varphi_{y_1}(\ga^2,0)\ga (\ga-\|b_1\|)\,. 
\end{aligned}
\end{equation}
Furthermore, $b_1\in (conv(g(\gO)))_1$ implies $\|b_1\|\le \max_{\go\in\gO}\|\go_1\|\psi(\|\go_1\|^2,\go_2).$ 
Thus, we arrive at the following result. 
\begin{corollary}\label{cor30} Suppose that $g(\go)=\go\,\psi(\|\go_1\|^2,\go_2)$ for some function $\psi\ge 0$ and $\mu_0(\go)=\mu_*(\|\go_1\|^2,\go_2)$ be such that  $\psi(\|\go_1\|^2,\go_2)\mu_*(\|\go_1\|^2,\go_2)$ is even in each coordinate of $\go_2.$ Let also $W(a)\ =\ \varphi(\|a_1\|^2,a_2)$ with $\varphi_{y_1}(\ga^2,0)>0$ and $\varphi(y_1,y_2)\le \varphi(y_1,0)$ for all $y_1,y_2.$ Define
\[
\ga\ \equiv\ \frac{\int_{0}^\infty r^{L_1} \int_{\R^{L_2}} \mu_*(r^2,\go_2)\psi(r^2,\go_2)d\go_2 dr}{\int_{0}^\infty r^{L_1-1} \int_{\R^{L_2}} \mu_*(r^2,\go_2)d\go_2 dr}\,.
\]
If
\begin{equation}\label{Dwor}
\max_{\|b_1\|\le \max_{\go\in\gO}\|\go_1\|\psi(\|\go_1\|^2,\go_2)}(\varphi(\|b_1\|^2,0)-\varphi(\ga^2,0)+2\varphi_{y_1}(\ga^2,0)\ga (\ga-\|b_1\|))\ \le\ 0\,, 
\end{equation}
then $a(\go)=\binom{\ga \go_1/\|\go_1\|}{0}$ is an optimal policy. The optimal policy is unique if $\varphi(y_1,y_2)<\varphi(y_1,0)$ for all $y_2\not=0$ and the maximum in \eqref{Dwor} is attained only when $\|b_1\|=\ga$.\footnote{Uniqueness follows from Proposition \ref{uniqueness}.}
\end{corollary}

It is straightforward to extend this analysis to the more general setup of Corollary \ref{tamura} with $W(a)=\varphi(a^\top H a)$ with an increasing $\varphi$ and $\mu_0(\go)=\mu_*(\go^\top \Sigma^{-1}\go)$, in which case $\Xi$ will be an ellipsoid whose Hausdorff dimension equals the number of positive eigenvalues of $H$ minus one (see  Theorem \ref{frostman}). It is also interesting to link these results to those of \cite{DworczakMartini2019}. Condition \eqref{Dwor} means that the graph of the function $\varphi(x^2)$ lies below its tangent at $x=\ga.$ When this condition is violated, one can consider the affine closure of $\varphi(x^2)$ as in  \cite{DworczakMartini2019}. In this case, the tangent will touch the graph of $\varphi(x^2)$ in several points $r_i$ and the optimal policy will be to project $\go_1$ onto one of the spheres $\|\go_1\|=r_i.$ It is then possible to extend the beautiful results of  \cite{DworczakMartini2019} to this nonlinear setting. 

We complete this section with an example where $\go$ takes values on a real analytic manifold in $\R^M$ . 
Namely, suppose that the sender observes the realization of $\go=(\go_1,\cdots,\go_M)$ of the probabilities of some states of the world, with $\sum_{i=1}^M\go_i=1.$ Note that in this case $\go$ lives on the unit simplex which is a real analytic manifold in $\R^M,$ but all our results directly apply in this setting as long as the prior is absolutely continuous with respect to the Lebesgue measure restricted to the unit simplex. We assume that $\mu_0(\go)$ is given by the Dirichlet distribution on the unit simplex,  
\begin{equation}\label{dirichlet} 
\mu_0(\go;\ga)\ =\ \frac1{B(\ga)}\prod_{i=1}^M\go_i^{\ga_i-1},\ B(\ga)\ =\ \frac{\prod_{i=1}^M \Gamma(\ga_i)}{\Gamma(\sum_{i=1}^M\ga_i)}\,.
\end{equation}
We will also be assuming a specific function form of the social welfare function that depends only the total probabilities of certain groups of states as well as on the relative entropy of the corresponding distributions. It is easy to micro-found such a welfare function in a setting with limited attention. See, e.g., \cite{gabaix2019behavioral}. 

\begin{corollary} \label{cor-dirichlet} Suppose that $\mu_0$ is the Dirichlet distribution on the unit simplex $\Delta_M$ with parameters $\ga=\binom{\bar\ga_1}{\bar\ga_2}$ and $\go=\binom{\bar\go_1}{\bar\go_2}$ with $\bar\ga_1,\bar\go_1\in \R_+^{M_1},\ \bar\ga_2,\bar\go_2\in \R_+^{M_2}.$ Let 
\[
g(\go)\ =\ \binom{\psi_1({\bf 1}^\top \bar\go_1)\bar \go_1}{\psi_2({\bf 1}^\top \bar\go_2)\bar \go_2}
\]
for some functions $\psi_i\ge 0, i=1,2,$ and 
\[
W(a)\ =\ \sum_{i=1}^2 (q_i\cE_i(\bar a_i)+\varphi_i({\bf 1}^\top \bar a_i))
\]
where $a=\binom{\bar a_1}{\bar a_2}\in \R^{M_1}\oplus \R^{M_2},$ with  $\cE_i(a)\ =\ \sum_j a(j) \log(a(j)/y_i(j))$ being the negative or the relative entropy, $y_i\in \R_+^{M_i}, i=1,2,$ are arbitrary vectors, $q_i>0,$ and $\varphi_i$ are arbitrary smooth functions. Define 
\[
\gamma_i\ =\ E^{\mu_0}[\psi_i({\bf 1}^\top \bar\go_i){\bf 1}^\top \bar\go_i],\ i=1,2,
\]
and 
\[
a(\go)\ =\ \binom{\gamma_1\,\bar\go_1/({\bf 1}^\top \bar\go_1)}{\gamma_2\,\bar\go_2/({\bf 1}^\top \bar\go_2)}\,.
\]
Suppose that 
\begin{equation}\label{max-cond-dirichlet}
\max_{0\le\bar b_i\le \max_{x\in [0,1]}\psi_i(x)x}\left(
\varphi_i(\bar b_i)-\varphi_i(\gamma_i)+\varphi_i'(\gamma_i)(\gamma-\bar b_i)\ -\ q_i(\bar b_i \log(\gamma_i/\bar b_i)-\gamma_i+\bar b_i)
\right)\ \le\ 0,\ i=1,\ 2\,. 
\end{equation}
Then, $a(\go)$ is an optimal policy. If for each $i=1,2$ the maximum in \eqref{max-cond-dirichlet} is attained only when $\bar b_i=\gamma_i$, then the optimal policy is unique. 
\end{corollary}

As in our discussion following Corollary \ref{cor30}, it is possible to show that when \eqref{max-cond-dirichlet} is violated, the optimal policy will be to project $\bar\go_i$ onto one of the multiple $\ell_1$-spheres defined by $\sum_j \bar a(j)=\gamma$ for several values of $\gamma$ corresponding to points where the affine closure touches the graph of the function in \eqref{max-cond-dirichlet}.

\section{Beyond Moment Persuasion}

Theorem \ref{frostman} implies that $D_{aa}W,$ the Hessian of $W$, is a key determinant of the structure of optimal policies. Is there an analog of $D_{aa}W$ for the more general  setting of Theorem \ref{mainth-limit}? What determines the natural convex and concave components of the problem? We do not have a complete answer to these questions. However, as we show in this section, something can be said in the case when the uncertainty is small. 

The structure of the optimal partition (Theorem \ref{regular-partition}) can be complex and non-linear.\footnote{In general, the boundaries of the sets $\gO_k$ might be represented by complicated hyper-surfaces, and some of $\gO_k$ might even feature multiple disconnected components.}  One may ask whether it is possible to ``linearize" these partitions, just as one can linearize equilibria in complex, non-linear economic models, assuming the deviations from the steady state are small. As we show below, this is indeed possible.

Everywhere in this section, we make the following assumption.
\begin{assumption}\label{ass-eps} There exists a small parameter $\eps$ such that the functions defining the equilibrium conditions, $G,$ and the welfare function, $W,$ are given by $G(a,\eps\go)$ and $W(a,\eps\go).$
\end{assumption}

Parameter $\eps$ has two interpretations. First, it could mean small deviations from a steady state (as is common in the literature on log-linear approximations). Second, $\eps$ could be interpreted as capturing the sensitivity of economic quantities to changes in $\go.$ In the limit when $\eps=0,$ equilibrium does not depend on shocks to $\go.$ We use $a^0=a_*(0)$ to denote this ``steady state" equilibrium. By definition, it is given by the unique solution to the system $G(a^0,0)\ =\ 0,$ and the corresponding social welfare is $W(a^0,0).$

\begin{assumption}[The information relevance matrix]\label{oonondeg}
We assume the matrix $\cD(0)$ with
\begin{equation}
\begin{aligned}
&\cD(\go)\ \equiv\ D_{\go\go} (W(a_*(\go),\go))\ -\ W_{\go\go}(a_*(\go),\go)
\end{aligned}
\end{equation}
is non-degenerate. We refer to $\cD$ as the information relevance matrix.
\end{assumption}

We are now ready to state the main result of this section, showing how the optimal linearized partition can be characterized explicitly in terms of the information relevance matrix $\cD$.

\begin{theorem}[Linearized partition]\label{mainLimit} Under the hypothesis of Theorem \ref{mainth1} and Assumptions \ref{ass-eps} and \ref{oonondeg}, let $\{\gO_k(\eps)\}_{k=1}^K$ be the corresponding optimal partition. Then, for any sequence $\eps_l\to 0,\ l>0,$ there exists a sub-sequence $\eps_{l_j},\ j>0,$  such that the optimal partition $\{\gO_k(\eps_{l_j})\}_{k=1}^K$ converges to an almost sure partition $\{\tilde \gO_k^*\}_{k=1}^K$ satisfying
\[
\tilde \gO_k^*\ =\ \{\go\in \gO:\ (M_1(k)-M_1(l))^\top \cD(0) \go\ >\ 0.5(M_1(k)^\top \,\cD(0) \, M_1(k)-M_1(l)^\top \,\cD(0) \, M_1(l))\ \forall\ l\not=k\}\,,
\]
where we have defined $M_1(k)\ \equiv\ E[\go|\tilde \gO_k^*]\,.$ In particular, for this limiting partition, each set $\tilde \gO_k^*$ is convex. If the matrix $\cD$ from Assumption \ref{oonondeg} is negative semi-definite, then all sets $\tilde\gO_k^*$ are empty except for one; that is, it is optimal to reveal no information.
\end{theorem}

Theorem \ref{mainLimit} implies that the general problem of optimal information design converges to a quadratic moment persuasion when $\eps$ is small. The matrix $\cD(0)$ of Assumption \ref{oonondeg} incorporates information both about the hessian of $H$ and about other partial derivatives of $G$. Moment persuasion setting corresponds to the case when $D_{\go a}G=0.$ One interesting effect we observe is that, in general, non-zero partial derivatives $D_{\go a}G$ may have a a major impact on the structure of the $\cD$ matrix. In particular, when $\eps$ is sufficiently small and $K$ is sufficiently large, we are in a position to apply Theorem \ref{frostman} and Corollary \ref{tamura}, linking the number of positive eigenvalues of $\cD$ to the dimension of the support of optimal policies.

\newpage

\bibliographystyle{aer}
\bibliography{bibliography}

\newpage

\appendix

{\bf \Huge Internet Appendix}

\section{Finite Partitions: Proofs}

\begin{proof}[Proof of Lemma \ref{existence}] First, by uniform monotonicity, the map 
\[
a\to\ F(a)= a+\gd E[G(a,\go)|k]
\]
is a contraction for sufficiently small $\gd.$ Indeed, by monotonicity, 
\[
\|F(a_1)-F(a_2)\|^2\ \le\ \|a_1-a_2\|^2-2\eps\gd\|a_1-a_2\|^2\ +\ \gd^2\eps^{-2} \|a_1-a_2\|^2\,. 
\]
As a result, there exists a unique equilibrium by the Banach fixed point theorem. Then, with $a=a(k),$
\begin{equation}
\begin{aligned}
&E[(a_*(\go)-a)^\top\,G(a,\go)|k]\ =\ E[(a_*(\go)-a)^\top\,(G(a,\go)-G(a_*(\go),\go))|k]\\ 
&\ge\ \eps\,E[\|a_*(\go)-a\|^2|k]\ \ge\ \eps(E[\|a_*(\go)\|^2+2\|a\|\|a_*(\go)\| |k]\ +\ \|a\|^2)\,. 
\end{aligned}
\end{equation}
At the same time, 
\begin{equation}
\begin{aligned}
&E[(a_*(\go)-a)^\top\,G(a,\go)|k]\ =\ E[a_*(\go)^\top\,G(a,\go)|k]\ \le\ \eps^{-1} E[\|a_*(\go)\|\,\|a-a_*(\go)\||k]\\
& =\ \eps^{-1}(E[\|a_*(\go)\|^2|k]+\|a\|E[\|a_*(\go)\||k])
\end{aligned}
\end{equation}
and the claim follows. 
\end{proof}

\begin{proof}[Proof of Theorem \ref{mainth1}] 
The fact that social welfare is bounded and depends smoothly on the information design follows by the same arguments as in the proof of Lemma \ref{lem-approx}.

Existence of an optimal information design then follows trivially from compactness. Indeed, since $\pi_k(\go)\in [0,1]$, the are square integrable and, hence, compact in the weak topology of $L_2(\mu_0).$ The identity $\sum_k \pi_k=1$ is trivially preserved in the limit. Continuity of social welfare in $\pi_k$ follows directly from the assumed integrability and regularity, hence the existence of an optimal design. 

The equilibrium conditions can be rewritten as
\begin{equation}
E_{\mu_s}[G(a(s),\go)|s]\ =\ 0\,.
\end{equation}
Here,
\[
\mu_s(\go)\ =\ \frac{\pi(s|\go)\mu_0(\go)}{\int \pi(s|\go)\mu_0(\go) d\go}
\]
and hence
\[
E_{\mu_k}[G(a(s),\go)]\ =\ \frac{\int \pi(k|\go)\mu_0(\go) G(a(k),\ \go)d\go}{\int \pi(k|\go)\mu_0(\go) d\go}\,.
\]
By assumption, equilibrium $a$ depends continuously on $\{\pi_k\}.$ Since the map
\[
(\{\pi_k\},\ \{a_k\})\ \to\ \left\{ \int \pi_k(\go)\mu_0(\go) G(a(k,\eps),\ \go)d\go\right\}
\]
is real analytic, and has a non-degenerate Jacobian with respect to $a,$ the assumed continuity of $a$ and the implicit function theorem imply that $a$ is in fact real analytic in $\{\pi_k\}.$ To compute the Frechet differentials of $a(s),$ we take a small perturbation $\eta(\go)$ of $\pi_k(\go)$. By the regularity assumption and the Implicit Function Theorem,
\[
a(k,\eps)\ =\ a(k)+\eps a^{(1)}(k)+0.5\eps^2a^{(2)}(k)\ +\ o(\eps^2)
\]
for some $a^{(1)}(k),\ a^{(2)}(k)\,.$ Let us rewrite
\begin{equation}
\begin{aligned}
&0\ =\ \int (\pi_k(\go)+\eps \eta(\go))\mu_0(\go) G(a(k,\eps),\ \go)d\go\\
&=\ \int (\pi_k(\go)+\eps \eta(\go))\mu_0(\go) G(a(k)+\eps a^{(1)}(k)+0.5\eps^2a^{(2)}(k),\ \go)d\go\\
&\approx\ \Bigg(
\int \pi_k(\go) \mu_0(\go)\Bigg(G(a(k),\ \go)+G_a (\eps a^{(1)}(k)+0.5\eps^2a^{(2)}(k))\\
&+0.5 G_{aa}(\eps a^{(1)}(k), \eps a^{(1)}(k))
\Bigg)d\go\\
&+\eps \int \eta(\go)\mu_0(\go) \Bigg(G(a(k))+G_a\eps a^{(1)}(k)
\Bigg)d\go
\Bigg)\\
&=\ \Bigg(\eps
\Bigg(
\int \pi_k(\go) \mu_0(\go)G_a a^{(1)}(k)d\go+\int \eta(\go)\mu_0(\go) G(a(k))d\go\Bigg)\\
&+0.5\eps^2
\Bigg(\int \pi_k(\go) \mu_0(\go)[G_a a^{(2)}(k)+G_{aa}(a(k),\go)( a^{(1)}(k), a^{(1)}(k))]d\go\\
&+2\int \eta(\go)\mu_0(\go) G_a(a(k),\go) a^{(1)}(k)d\go\Bigg)
\Bigg)
\end{aligned}
\end{equation}
As a result, we get
\begin{equation}\label{x1}
\begin{aligned}
&a^{(1)}(k)\ =\ -\bar G_a(k)^{-1}\,\int \eta(\go)\mu_0(\go) G(a(k),\go)d\go,\ \bar G_a(k)\ =\ \int \pi_k(\go) \mu_0(\go)G_a d\go\,,
\end{aligned}
\end{equation}
while
\begin{equation}\label{x2}
\begin{aligned}
&a^{(2)}(k)\ =\ -\bar G_a(k)^{-1}\,\Bigg(
\int \pi_k(\go) \mu_0(\go)G_{aa}(a(k),\go)( a^{(1)}(k), a^{(1)}(k))d\go\\
& +\ 2\int \eta(\go)\mu_0(\go) G_a(a(k),\go) a^{(1)}(k)d\go
\Bigg)\,.
\end{aligned}
\end{equation}
Consider the social welfare function
\begin{equation}
\begin{aligned}
&\bar W(\pi)\ =\ E[W(a(s),\go)]\ =\ \sum_k \int_\gO W(a(k),\go)\pi_k(\go)\mu_0(\go)d\go\,.
\end{aligned}
\end{equation}
Suppose that the optimal information structure is not a partition. Then, there exists a subset $I\subset\gO$ of positive $\mu_0$-measure and an index $k$ such that $\pi_k(\go)\in (0,1)$ for $\mu_0$-almost all $\go\in I.$ Since $\sum_i\pi_i(\go)=1$ and $\pi_i(\go)\in[0,1],$ there must be an index $k_1\not=k$ and a subset $I_1\subset I$ such that $\pi_{k_1}(\go)\in (0,1)$ for $\mu_0$-almost all $\go\in I_1.$ Consider a small perturbation $\{\tilde\pi(\eps)\}_i$ of the information design, keeping $\pi_i,\ i\not=k,k_1$ fixed and changing $\pi_k(\go)\to \pi_k(\go)+\eps \eta(\go),\ \pi_{k_1}(\go)\to\pi_{k_1}(\go)-\eps(\go)$ where $\eta(\go)$ in an arbitrary bounded function with $\eta(\go)=0$ for all $\go\not\in I_1.$ Define $\eta_k(\go)=\eta(\go),\ \eta_{k_1}(\go)=-\eta(\go),$ and $\eta_i(\go)=0$ for all $i\not=k,k_1.$ A second-order Taylor expansion in $\eps$ gives
\begin{equation}\label{w-expan1}
\begin{aligned}
&\sum_{i} \int_\gO W(a(i,\eps),\go) (\pi_i(\go)+\eps \eta_i(\go))\mu_0(\go)d\go\\
&\approx\  \int_\gO \Bigg(W(a(i),\go)+W_a(a(i),\go)(\eps a^{(1)}(i)+0.5\eps^2a^{(2)}(i))\\
&+0.5 W_{aa}(a(i),\go)\eps^2( a^{(1)}(i), a^{(1)}(i))\Bigg) (\pi_i(\go)+\eps \eta_i(\go))\mu_0(\go)d\go\\
&=\ \bar W(\pi)\ +\ \eps\sum_i \Bigg(
\int_\gO (W(a(i),\go) \eta_i(\go)+ W_a(a(i),\go)a^{(1)}(i)\pi_i(\go))\mu_0(\go)d\go
\Bigg)\\
&+0.5\eps^2
 \sum_i \int_\gO\Big( W_{aa}(a(i),\go)( a^{(1)}(i), a^{(1)}(i))\pi_i(\go)\\
 &+W_a(a(i),\go) a^{(2)}(i) \pi_i(\go)
 +W_a(a(i),\go) a^{(1)}(i)\eta_i(\go)
 \Big) \mu_0(\go)d\go
\end{aligned}
\end{equation}
Since, by assumption, $\{\pi_i\}$ is an optimal information design, it has to be that the first order term in \eqref{w-expan1} is zero, while the second-order term is always non-positive. We can rewrite the first order term as
\begin{equation}\label{w-expan}
\begin{aligned}
&\sum_i \Bigg(
\int_\gO (W(a(i),\go) \eta_i(\go)+ W_a(a(i),\go)a^{(1)}(i)\pi_i(\go))\mu_0(\go)d\go
\Bigg)\\
&=\ \sum_i \int_\gO \Bigg(W(a(i),\go)\\
& -\  \Big(\int W_a(a(i),\go_1)\pi_i(\go_1)\mu_0(\go_1)d\go_1
\Big)\bar G_a(i)^{-1}\, G(a(i),\go)
\Bigg)
\eta_i(\go)\mu_0(\go)d\go
\end{aligned}
\end{equation}
and hence it is zero for all considered perturbations if and only if
\begin{equation}\label{w-expan2}
\begin{aligned}
&W(a(k),\go)\ -\ \Big(\int W_a(a(k), \go_1)\pi_k(\go_1)\mu_0(\go_1)d\go_1
\Big)\bar G_a(k)^{-1}\, G(a(k),\go)\\
&=\ W(a(k_1),\go)\ -\ \Big(\int W_a(a(k_1),\go)\pi_{k_1}(\go_1)\mu_0(\go_1)d\go_1
\Big)\bar G_a(k_1)^{-1}\, G(a(k_1),\go)
\end{aligned}
\end{equation}
Lebesgue-almost surely for $\go\in I_1.$ By Proposition \ref{zero-go}, \eqref{w-expan2} also holds for all $\go\in\gO.$ Hence, by Assumption \ref{main-ass-indep}, $a(k)=a(k_1),$ which contradicts our assumption that all $a(k)$ are different. 
\end{proof}

\begin{proof}[Proof of Theorem \ref{regular-partition}]
Suppose a partition $\go\ =\ \cup_k \gO_k$ is optimal. By regularity, equilibrium actions satisfy the first order conditions
\[
\int_{\gO_k} G(a(k),\ \go)\mu_0(\go)d\go\ =\ 0\,.
\]
Consider a small perturbation, whereby we move a small mass on a set $\cI\subset \gO_k$ to $\gO_l.$ Then, the marginal change in $a_n(k)$ can be determined from
\begin{equation}
\begin{aligned}
&0\ =\ \int_{\gO_k} G(a(k),\ \go)\mu_0(\go)d\go\ -\ \int_{\gO_k\setminus\cI} G(a(k,\cI),\ \go)\mu_0(\go)d\go\\
&\approx\ -\int_{\gO_k} D_aG(a(k),\ \go)\Delta a(k)\ \mu_0(\go)d\go\ + \int_{\cI} G(a(k),\ \go)\mu_0(\go)d\go\,,
\end{aligned}
\end{equation}
implying that the first order change in $a$ is given by
\[
\Delta a(k)\ \approx\ (\bar D_aG(k))^{-1}\int_{\cI} G(a(k),\ \go)\mu_0(\go)d\go\,.
\]
Thus, the change in welfare is\footnote{Note that $D_aW$ is a horizontal (row) vector.}
\begin{equation}
\begin{aligned}
&\Delta W\ =\ \int_{\gO_k}W(a(k),\ \go)\mu_0(\go)d\go\ -\ \int_{\gO_k\setminus \cI}W(a(k,\cI),\ \go)\mu_0(\go)d\go\\
&+\int_{\gO_l}W(a(l),\ \go)\mu_0(\go)d\go\ -\ \int_{\gO_l\cup \cI}W(a(l,\cI),\ \go)\mu_0(\go)d\go\\
&\approx\ -\int_{\gO_k}D_aW(a(k),\ \go)\Delta a(k)\mu_0(\go)d\go+\int_{\cI}W(a(k),\ \go)\mu_0(\go)d\go\\
&-\int_{\gO_l}D_aW(a(l),\ \go)\Delta a(l)\mu_0(\go)d\go-\int_{\cI}W(a(l),\ \go)\mu_0(\go)d\go\\
&=\  -\bar D_aW(k)(\bar D_aG(k))^{-1}\int_{\cI} G(a(k),\ \go)\mu_0(\go)d\go+\int_{\cI}W(a(k),\ \go)\mu_0(\go)d\go\\
&+\bar D_aW(l)(\bar D_aG(l))^{-1}\int_{\cI} G(a(l),\ \go)\mu_0(\go)d\go-\int_{\cI}W(a(l),\ \go)\mu_0(\go)d\go\,.
\end{aligned}
\end{equation}
This expression has to be non-negative for any $\cI$ of positive Lebesgue measure. Thus,
\begin{equation}
\begin{aligned}
& -\bar D_aW(k)(\bar D_aG(k))^{-1}G(a(k),\ \go)+W(a(k),\ \go)\\
&+\bar D_aW(l)(\bar D_aG(l))^{-1} G(a(l),\ \go)\ -\ W(a(l),\ \go)\ \ge\ 0
\end{aligned}
\end{equation}
for Lebesgue almost any $\go\in\gO_k.$
\end{proof}

\begin{proof}[Proof of Proposition \ref{main convexity}] First, we note that $y=\binom{y_1}{y_2}\in \hat g(\gO_k\cap X)$ where $y_1\in \R^L$ if and only if 
\[
W(a(k))-x_k^\top (a(k)-y_1)\ =\ \max_{1\le l\le K} (W(a(l))\ -\ x_l^\top (a(l)-y_1))
\]
and $y_1\in \hat g(X).$ Both sets are convex and hence so is their intersection. To show monotonicity of $D_aW(a(\hat g^{-1}(y)))$, pick a $y,z$ such that $y,y+z\in g(\gO_k\cap X).$ By convexity, $y+tz \in \hat g(\gO_k\cap X)$ for all $t\in [0,1].$ Our goal is to show that 
\[
(D_aW(a(g^{-1}(y+z)))-D_aW(a(g^{-1}(y)))z\ \ge\ 0\,. 
\]
Since $a$ is constant inside each $\gO_k,$ it suffices to show this inequality when $y$ and $y+z$ are infinitesimally close to the boundary between two regions, $\gO_{k_1}$ and $\gO_{k_2}.$ Let $y$ belong to that boundary and $y+\eps z\in \gO_{k_2}$. Then, 
\[
W(a(k_2))-D_aW(a(k_2))(a(k_2)-(y+\eps z))\ \ge\ W(a(k_1))-D_aW(a(k_1))(a(k_1)-(y+\eps z))
\]
and 
\[
W(a(k_2))-D_aW(a(k_2))(a(k_2)-y)\ =\ W(a(k_1))-D_aW(a(k_1))(a(k_1)-y)
\]
Subtracting, we get the required monotonicity. 
\end{proof}

\section{Finite Partitions: The Small Uncertainty Limit}

When the policy-maker sends signal $k$, the receivers learn that $\go\in \gO_k$. As a result, the receivers' posterior estimate of the conditional mean of $\go$ is then given by
\begin{equation}\label{M1}
M_1(\gO_k)\ \equiv\ E[\go|\go\in \gO_k]\ = \frac{\int_{\gO_k} \go\mu_0(\go)d\go}{P_k}\ \in\ \R^m\,,
\end{equation}
where
\[
P_k\ =\ \cP(k)\ =\ \int_{\gO_k}\mu_0(\go)d\go\,.
\]


Define
\begin{equation}
\cG\ \equiv\ (D_aG(a^0,0))^{-1}D_\go G(a^0,0)\ \in \R^{M\times L}\,.
\end{equation}

The following lemma follows by direct calculation.
\begin{lemma}\label{ak1} For any sequence $\eps_\nu \to 0,\ \nu\in \Z_+,$ there exists a sub-sequence $\eps_{\nu_j},\ j>0,$ such that the optimal partitions $\{\gO_k(\eps_{\nu_j})\}_{k=1}^K$ converge to a limiting partition $\{\gO_k(0)\}_{k=1}^K\,$ as $j\to\infty.$  In this limit,
\begin{equation}
a_k(\eps_{\nu_j})\ =\ a^0_k\ -\ \eps_{\nu_j} \cG\,M_1(\gO_k(0))\ +\ o(\eps_{\nu_j})\,.
\end{equation}
\end{lemma}
Lemma \ref{ak1} provides an intuitive explanation for the role of the matrix $\cG.$ Namely, in the linear approximation, the receivers' action is given by a linear transformation of $E[\go|k],$ the first moment of $\go$ given the signal: $a_k\ \approx\ a^0_k\ -\ \cG E[\go|k].$ Thus, the matrix $-\cG$ captures how strongly receivers' actions respond to changes in beliefs.

\begin{proof}[Proof of Lemma \ref{ak1}]  Trivially, the set of partitions is compact and hence we can find a subsequence $\{\gO_k(\eps_j)\}$ converging to some partition $\{\gO_k(0)\}$ in the sense that their indicator functions converge in $L_2.$ We have
\begin{equation}
\begin{aligned}
&0\ =\ \int_{\gO_k(\eps)} G(a(k,\eps),\eps\go)\mu_0(\go)d\go\ =\ \int_{\tilde\gO_k(\eps)} G(a(k,\eps),\go)\mu_0(\go)d\go
\end{aligned}
\end{equation}
Now,
\begin{equation}\label{id-b}
\begin{aligned}
&0\ =\ \int_{\gO_k(\eps)} G(a(k,\eps),\eps\go)\mu_0(\go)d\go\\
&=\ G(a(k,\eps),0)M(\gO_k(\eps))\ +\ \eps\,D_\go G(a(k,\eps),0)\,M_1(\gO_k(\eps))\ +\ O(\eps^2)\,.
\end{aligned}
\end{equation}
Let us show that $a(k,\eps)-a(k,0)\ =\ O(\eps).$ Suppose the contrary. Then there exists a sequence $\eps_m\to0$ such that $\|a(k,\eps)-a(k,0)\|\eps^{-1}\ \to\infty.$ We have
\begin{equation}\label{id-b1}
\begin{aligned}
&G(a(k,\eps),0)\ -\ G(a(k,0),0)\ =\ \int_0^1D_aG(a(k,0)+t(a(k,\eps)-a(k,0)))(a(k,\eps)-a(k,0))dt\\
&\ge\ c \|a(k,\eps)-a(k,0)\|
\end{aligned}
\end{equation}
for some $c>0$ due to the continuity and non-degeneracy of $D_aG(0)=D_aG(a(k,0)).$ Dividing \eqref{id-b} by $\eps,$ we get a contradiction.

Define
\[
 a^{(1)}(k)\ \equiv\ -D_aG(0)^{-1}D_\go G(a(k),0)M_1(\gO_k(0))\ =\ -\cG\,M_1(\gO_k(0))\,.
\]
Let us now show that $a(k,\eps)-a(k,0)\ =\ \eps a^{(1)}(k)\ +\ o(\eps).$ Suppose the contrary. Then, $\|\eps^{-1}(a(k,\eps)-a(k,0))-a^{(1)}(k)\|\ >\ c$ for some $c>0$ along a sequence of $\eps\to0.$ By \eqref{id-b1},
\begin{equation}\label{id-b2}
\begin{aligned}
&0\ =\ \int_{\gO_k(\eps)} G(a(k,\eps),\eps\go)\mu_0(\go)d\go\\
&=\ G(a(k,\eps),0)M(\gO_k(\eps))\ +\ \eps\,D_\go G(a(k,\eps),0)\,M_1(\gO_k(\eps))\ +\ O(\eps^2)\\
&=\ \eps D_aG(0) \eps^{-1}(a(k,\eps)-a(k,0))M(\gO_k(\eps))\ +\ \eps\,D_\go G(a(k),0)\,M_1(\gO_k(\eps))\ +\ O(\eps^2)\,,
\end{aligned}
\end{equation}
and we get a contradiction taking the limit as $\eps\to0.$
\end{proof}

\begin{proof}[Proof of Theorem \ref{mainLimit}] We have 
\begin{equation}\label{set11}
\begin{aligned}
&\gO_k(\eps)\ =\  \{\go\in\gO:\  -\bar D_aW(k,\eps)(\bar D_aG(k,\eps))^{-1}G(a(k,\eps),\ \eps\go)+W(a(k,\eps),\ \eps\go)\\ >\ 
&-\bar D_aW(l,\eps)(\bar D_aG(l,\eps))^{-1} G(a(l,\eps),\ \eps\go)\ +\ W(a(l,\eps),\ \eps\go)\ \forall\ l\not=k.\}
\end{aligned}
\end{equation}
The proof of the theorem is based on the following technical lemma. 

\begin{lemma}\label{pert} We have 
\begin{equation}
\begin{aligned}
& -\bar D_aW(k,\eps)(\bar D_aG(k,\eps))^{-1}G(a(k,\eps),\ \eps\go)+W(a(k,\eps),\ \eps\go)\\
&=\ W(0)\ -\ 0.5M_1(k)^\top \cD \eps^2 \go\ +\ 0.5\eps^2 M_1(k)^\top \,\cD \, M_1(k)\ +\ \eps W_\go(0)\go\\
&+0.5\eps^2\go^\top W_{\go\go}(0)\go\ -D_aW(0) D_aG(0)^{-1}(D_\go G(0)\go +0.5\go^\top G_{\go\go}(0)\go) +\ o(\eps^2)\,. 
 \end{aligned}
 \end{equation}
\end{lemma}
\begin{proof}
We have
\begin{equation}
\begin{aligned}
&\bar D_aW(k,\eps)\ =\ \int_{\gO_k(\eps)}D_aW(a(k,\eps),\eps\go)\mu_0(\go)d\go\\
& =\ \int_{\gO_k(\eps)}(D_aW(0)+\eps \go^\top D_\go W(0)^\top+\eps a^{(1)}(k)^\top D_{aa}^2W(0)+o(\eps))\mu_0(\go)d\go\\
&\ =\ (D_aW(0)+\eps (a^{(1)}(k))^\top D_{aa}^2W(0) \ +\ \eps\, M_1(\gO_k(0))^\top (D_\go W(0))^\top )M(\gO_k(\eps))\,\ +\ o(\eps)\ \in \R^{1\times M}\,.
\end{aligned}
\end{equation}
At the same time, an analogous calculation implies that
\begin{equation}
\begin{aligned}
&\bar D_aG(k,\eps)\ =\ (D_aG(0)+\eps (a^{(1)}(k))^\top D_{aa}^2G(0) \ +\ \eps\, D_\go G(0)M_1(\gO_k(0)))M(\gO_k(\eps))\,\ +\ o(\eps)
\end{aligned}
\end{equation}
Here, $D_aG(0)=(\partial G_i/\partial a_j)$ and
\[
(D_\go G(0)M_1(\gO_k(0)))_{i,j}\ =\ \sum_k\frac{\partial^2 G_i}{\partial a_j\partial\go_k}M_{1,k}\,,
\]
and, similarly,
\[
((a^{(1)}(k))^\top D_{aa}^2G(0))_{i,j}\ =\ \sum_l (a^{(1)}(k))_l \frac{\partial^2 G_l}{\partial a_i\partial a_j}\ \in \R^{M\times M}\,.
\]
Thus, 
\begin{equation}
\begin{aligned}
&M(\gO_k(\eps))\bar D_aG(k,\eps)^{-1}\\
&=\ D_aG(0)^{-1}-D_aG(0)^{-1}\eps \Big((a^{(1)}(k))^\top D_{aa}^2G(0) \ +\ \eps\,  D_\go G(0)M_1(\gO_k(0))\Big)D_aG(0)^{-1}\ +\ o(\eps)\,,
\end{aligned}
\end{equation}
and therefore
\begin{equation}
\begin{aligned}
&\bar D_aW(k,\eps)(\bar D_aG(k,\eps))^{-1}\ =\ D_aW(0) D_aG(0)^{-1}\\
&+\ \eps (M_1^\top D_\go W(0)^\top D_aG(0)^{-1}\,+\, (a^{(1)}(k))^\top D_{aa}^2W(0)D_aG(0)^{-1})\\
&-\eps D_aW(0)D_aG(0)^{-1} \Big((a^{(1)}(k))^\top D_{aa}^2G(0) \ +\  D_\go G(0)M_1(\gO_k(0))\Big)D_aG(0)^{-1} +\ o(\eps)\\
&=\ D_aW(0) D_aG(0)^{-1}\\
&+\ \eps (M_1^\top D_\go W(0)^\top D_aG(0)^{-1}\,- M_1^\top\cG^\top D_{aa}^2W(0)D_aG(0)^{-1})\\
&-\eps D_aW(0)D_aG(0)^{-1} \Big(- M_1^\top\cG^\top D_{aa}^2G(0) \ +\ D_\go G(0)M_1\Big)D_aG(0)^{-1} +\ o(\eps)\\
&=\ D_aW(0) D_aG(0)^{-1}+\eps\Gamma+o(\eps)\,,
\end{aligned}
\end{equation}
where
\begin{equation}
\begin{aligned}
&\Gamma\ =\ M_1^\top D_\go W(0)^\top D_aG(0)^{-1}\\
&- M_1^\top\cG^\top D_{aa}^2W(0)D_aG(0)^{-1}-D_aW(0)D_aG(0)^{-1} \Big(- M_1^\top\cG^\top D_{aa}^2G(0) \ +\ D_\go G(0)M_1\Big)D_aG(0)^{-1}\,.
\end{aligned}
\end{equation}
Define
\[
 \tilde a^{(1)}(k,\eps)\ \equiv\ \eps^{-1}(a(k,\eps)-a(k,0))\ =\ a^{(1)}(k)\ +\ o(1)\,.
\]
Let also
\begin{equation}
\begin{aligned}
&G^{(2)}(k)\ \equiv\ 0.5\eps^2(a^{(1)}(k)^\top D_{aa}^2G(0)a^{(1)}(k) +\  2\go^\top D_\go G(0) a^{(1)}(k)\ +\ \go^\top G_{\go\go}(0)\go)\,,
\end{aligned}
\end{equation}
so that
\[
G(a(k,\eps),\ \eps\go)\ -\ (\eps D_aG(0) \tilde a^{(1)}(k,\eps)+\eps D_\go G(0)\go)\ =\ \eps^2G^{(2)}(k)\ +\ o(\eps^2)\,,
\]
where we have used that $G(0)=0.$ While we cannot prove that $\eps (\tilde a^{(1)}(k)-a^{(1)}(k))=o(\eps^2),$ we show that this term cancels out. We have
\begin{equation}
\begin{aligned}
&-\bar D_aW(k,\eps)(\bar D_aG(k,\eps))^{-1}G(a(k,\eps),\ \eps\go)+W(a(k,\eps),\ \eps\go)\\
&\approx\ -\bar D_aW(k,\eps)(\bar D_aG(k,\eps))^{-1}\Big(\eps D_aG(0) \tilde a^{(1)}(k,\eps)+\eps D_\go G(0)\go +\eps^2G^{(2)}(k)\ +\  o(\eps^2)\Big)\\
&+\Bigg(W(0)+\eps D_aW(0) \tilde a^{(1)}(k,\eps)+\eps W_\go(0)\go \\
&+0.5\eps^2
\Big(
(a^{(1)}(k))^\top  D_{aa}^2W(0) a^{(1)}(k)+ \go^\top  W_{\go,\go}(0)\go+2 (a^{(1)}(k))^\top D_\go W(0)  \go
\Big)\ +\ o(\eps^2)
\Bigg)\\
&=\ -
\Big(
D_aW(0) D_aG(0)^{-1}\ +\ \eps \Gamma+o(\eps)\Big)\\
&\times \Big(\eps D_aG(0) \tilde a^{(1)}(k,\eps)+\eps D_\go G(0)\go +\eps^2G^{(2)}(k)\ +\  o(\eps^2)\Big)\\
&+\Bigg(W(0)+\eps D_aW(0) \tilde a^{(1)}(k)+\eps W_\go(0)\go \\
&+0.5\eps^2
\Big(
(a^{(1)}(k))^\top  D_{aa}^2W(0) a^{(1)}(k)+ \go^\top  W_{\go,\go}(0)\go+2( a^{(1)}(k))^\top  D_\go W(0) \go
\Big)+o(\eps^2)
\Bigg)\\
\end{aligned}
\end{equation}
\begin{equation}
\begin{aligned}
&=\ W(0)\\
&+\ \eps
\Bigg(
-D_aW(0) D_aG(0)^{-1}
\Big(D_aG(0) \tilde a^{(1)}(k,\eps)+D_\go G(0)\go\Big)\ +
D_aW(0) \tilde a^{(1)}(k,\eps)+ W_\go(0)\go
\Bigg)\\
&+\eps^2
\Bigg(-D_aW(0) D_aG(0)^{-1}G^{(2)}(k)
- \Gamma  \Big(D_aG(0) a^{(1)}(k)+D_\go G(0)\go\Big)\\
&+0.5
\Big(
(a^{(1)}(k))^\top  D_{aa}^2W(0) a^{(1)}(k)+ \go^\top  W_{\go,\go}(0)\go+2  a^{(1)}(k)^\top  D_\go W(0)\go
\Big)
\Bigg)\ +\ o(\eps^2)\,.
\end{aligned}
\end{equation}
Thus, the terms with $\tilde a^{(1)}(k,\eps)$ have cancelled out. 
We have
\begin{equation}
\begin{aligned}
&\Gamma  \Big(D_aG(0) a^{(1)}(k)+D_\go G(0)\go\Big)\ =\ \Big(
M_1^\top D_\go W(0)^\top D_aG(0)^{-1}\\
&- M_1^\top\cG^\top D_{aa}^2W(0)D_aG(0)^{-1}-D_aW(0)D_aG(0)^{-1} \Big(- M_1^\top\cG^\top D_{aa}^2G(0) \ +\   D_\go G(0)M_1\Big)D_aG(0)^{-1}
\Big)\\
&\times  D_\go G(0)(\go-M_1)\\
&=\  \Big(
M_1^\top D_\go W(0)^\top - M_1^\top\cG^\top D_{aa}^2W(0)-D_aW(0)D_aG(0)^{-1} \Big(- M_1^\top\cG^\top D_{aa}^2G(0) \ +\  D_\go G(0)M_1\Big)\Big)\\
&\times  \cG(\go-M_1)\ =\ M_1^\top \cD_1\cG(\go-M_1)\,,
\end{aligned}
\end{equation}
where
\[
\cD_1\ =\ D_\go W(0)^\top-D_aW(0)D_aG(0)^{-1} D_\go G(0)-(\cG^\top D_{aa}^2W(0)-\cG^\top D_aW(0)D_aG(0)^{-1} D_{aa}^2G(0))\in \R^{L\times M}
\]
and where the three-dimensional tensor multiplication is understood as follows:
\begin{equation}
\begin{aligned}
&M_1^\top D_aW(0)D_aG(0)^{-1} D_\go G(0)\ =\ \sum_k M_{1,k}D_aW(0)D_aG(0)^{-1} D_aG_{\go_k}(0)\\
&M_1^\top \cG^\top  D_aW(0)D_aG(0)^{-1} D_\go G(0)\ =\ \sum_k (\cG M_{1})_k D_aW(0)D_aG(0)^{-1} D_aG_{a_k}(0)\,.
\end{aligned}
\end{equation}
Rewriting, we get
\begin{equation}
\begin{aligned}
&W(0)\ +\ \eps
\Bigg(
-D_aW(0) D_aG(0)^{-1}D_\go G(0)\go+ W_\go(0)\go
\Bigg)\\
&+\eps^2
\Bigg(-D_aW(0) D_aG(0)^{-1}\eps^2G^{(2)}(k)-M_1^\top \cD_1\cG(\go-M_1)
\\
&+0.5
\Big(
(a^{(1)}(k))^\top  D_{aa}^2W(0) a^{(1)}(k)+ \go^\top  W_{\go,\go}(0)\go+2  a^{(1)}(k)^\top  D_\go W(0)\go
\Big)
\Bigg)\ +\ o(\eps^2)\\
&=\ W(0)\ +\ \eps
\Bigg(
-D_aW(0) \cG\go+ W_\go(0)\go
\Bigg)\\
&+\eps^2
\Bigg(-D_aW(0) D_aG(0)^{-1}\eps^2G^{(2)}(k)-M_1^\top \cD_1\cG(\go-M_1)\\
&+0.5
\Big(
M_1^\top \cG^\top D_{aa}^2W(0)\cG M_1+ \go^\top  W_{\go,\go}(0)\go-2  (\cG M_1)^\top  D_\go W(0)\go
\Big)
\Bigg)\ +\ o(\eps^2)\,.
\end{aligned}
\end{equation}
Now,
\[
\eps^2G^{(2)}(k)\ =\ 0.5(M_1^\top \cG^\top D_{aa}^2G(0)\cG M_1\ -\  2(\cG M_1)^\top D_\go G(0)\go +\go^\top G_{\go\go}(0)\go )\,.
\]
Thus, the desired expression is given by
\begin{equation}
\begin{aligned}
&W(0)\ +\ \eps
\Bigg(
-D_aW(0) \cG\go+ W_\go(0)\go
\Bigg)\ +\ \eps^2(0.5 M_1^\top\cA M_1+M_1^\top\cB \go\ +\ \go^\top \mathcal C \go)
\end{aligned}
\end{equation}
where we have defined
\begin{equation}
\begin{aligned}
&\cA\ \equiv\ - D_aW(0) D_aG(0)^{-1}\cG^\top D_{aa}^2G(0)\cG+2\cD_1\cG+\cG^\top D_{aa}^2W(0)\cG \\
&=\  - D_aW(0) D_aG(0)^{-1}\cG^\top D_{aa}^2G(0)\cG+2\Big(D_\go W(0)^\top-D_aW(0)D_aG(0)^{-1} D_\go G(0)\\
&-(\cG^\top D_{aa}^2W(0)-D_aW(0)D_aG(0)^{-1}\cG^\top D_{aa}^2G(0))\Big)\cG+\cG^\top D_{aa}^2W(0)\cG\\
&=\ \cG^\top D_aW(0) D_aG(0)^{-1} D_{aa}^2G(0)\cG-\cG^\top D_{aa}^2W(0)\cG\\
&+2(D_\go W(0)^\top\cG -\cG^\top D_aW(0)D_aG(0)^{-1} D_\go G(0))\in \R^{L\times L}\\
&\cB\ \equiv\  \cG^\top D_aW(0) D_aG(0)^{-1} D_\go G(0)-\cD_1\cG-\cG^\top D_\go W(0)\\
&=\  \cG^\top D_aW(0) D_aG(0)^{-1} D_aG^\top_\go(0)-\Big(\cG^\top D_\go W(0)-\cG^\top D_aW(0)D_aG(0)^{-1} D_\go G(0)\\
&-(\cG^\top D_{aa}^2W(0)\cG -D_aW(0)D_aG(0)^{-1}\cG^\top D_{aa}^2G(0))\cG\Big)-D_\go W(0)^\top \cG
\end{aligned}
\end{equation}
Here, the first term is given by
\[
(D_aW(0) D_aG(0)^{-1} D_aG^\top_\go(0))_{i,j}\ =\ \sum_k (D_aW(0) D_aG(0)^{-1} )_k \frac{\partial^2 G_k}{\partial a_i\partial \go_j}\,,
\]
and the claim follows by a direct (but tedious) calculation. 
\end{proof}

The desired convergence is then a direct consequence of Lemma \ref{pert}. Indeed, by compactness, we can pick a converging sub-sequence and Lemma \ref{pert} implies that, in the limit, a point $\go$ satisfies inequalities \eqref{set11} if and only if $\go\in \tilde \gO_k^*$. 
\end{proof}

\section{The Unconstrained Problem}

\begin{proof}[Proof of Lemma \ref{lem-approx}] The proof requires some additional arguments because $\gO$ is not necessarily compact. 
First, consider an increasing sequence of compact sets $X_n=\{\go:g(\go)\le n\}$ such that $X_n$ converge to $\gO$ as $n\to\infty.$  For any measure $\mu,$ let $\mu_X$ be its restriction on $X.$ Let $a_n=a(\mu_{X_n}).$ The first observation is that Assumptions \ref{integrability} and \ref{ac} imply that $a_n\to a$ uniformly as $n\to \infty$. Indeed, 
\[
\int_{X_n} G(a_n,\go)d\mu(\go)\ =\ \int_{\gO} G(a,\go)d\mu(\go)=0
\]
implies that 
\begin{equation}
\begin{aligned}
&\int_{X_n} (G(a_n,\go)-G(a,\go))d\mu(\go)\ =\ \int_{\gO\setminus X_n} G(a,\go)d\mu(\go)\\
&\le\ \int_{\gO\setminus X_n} \eps^{-1}\|a-a_*(\go)\|d\mu(\go)\ \le\ 2\eps^{-1}\mu(\gO\setminus X_n)^{1/2}\left(\int_{\gO\setminus X_n} \|a_*(\go)\|^2d\mu(\go)\right)^{1/2}\\
&\ \le\ \eps^{-1}(\mu(\gO\setminus X_n)+\int_{\gO\setminus X_n} \|a_*(\go)\|^2d\mu(\go))\,.
\end{aligned}
\end{equation}
Multiplying by $(a-a_n)$, we get 
\[
\eps\,\|a-a_n\|^2 (1-\mu(\gO\setminus X_n))\ \le\  \|a-a_n\| \eps^{-1}(\mu(\gO\setminus X_n)+\int_{\gO\setminus X_n} \|a_*(\go)\|^2d\mu(\go))
\]
Furthermore, by Lemma \ref{existence}, $\|a-a_n\|\le\ 2\left(\int_{\gO} g(\go)d\mu(\go)\right)^{1/2}\ \le\ 1+\int_{\gO} g(\go)d\mu(\go)$ and therefore 
\[
\|a-a_n\|\ \le\ C\Bigg( \mu(\gO\setminus X_n)(1+\int_{\gO} g(\go)d\mu(\go))+\int_{\gO\setminus X_n} g(\go)d\mu(\go)
\Bigg)
\] 
for some constant $C>0.$ Now, pick a $\tau\in \Delta(\Delta(\gO)).$ Since the function $q(x)={\bf 1}_{x>n}$ is monotone increasing in $x$, we get 
\[
\mu(\gO\setminus X_n)\int_{\gO} g(\go)d\mu(\go)\ =\ \int_{\gO} q(g(\go))d\mu(\go) \int_{\gO} g(\go)d\mu(\go)\ \le\ \int_{\gO} q(g(\go))g(\go)d\mu(\go)\ =\ \int_{\gO\setminus X_n} g(\go)d\mu(\go)\,
\]
and therefore 
\[
\|a-a_n\|\ \le\ C\int_{\gO\setminus X_n}(1+2g(\go))d\mu(\go)\,.
\]
Then, we have by the Jensen inequality that 
\begin{equation}
\begin{aligned}
&|\bar W(\mu)-\bar W(\mu_{X_n})|\ \le\ \int_{\gO\setminus X_n} |W(a(\mu),\go)|d\mu(\go)\ +\ \int_{X_n}|W(a(\mu),\go)-W(a_n(\mu),\go)|d\mu(\go)\\
&\ \le\ \int_{\gO\setminus X_n}(g(\go) f(\int_{\gO} g(\go)d\mu(\go)))d\mu(\go)\\
&+\ \int_{\gO}\|a(\mu)-a_n(\mu)\| (g(\go) f(\int_{\gO} g(\go)d\mu(\go)))d\mu(\go)\\
&\ \le\ \int_{\gO\setminus X_n}g(\go)d\mu(\go)\,\int_\gO f(g(\go))d\mu(\go)\\
&+\ \|a(\mu)-a_n(\mu)\|  \int_{\gO} g(\go)d\mu(\go)\,\int_\gO f(g(\go))d\mu(\go)\,.
\end{aligned}
\end{equation}
Since the function $q(x)=x{\bf 1}_{x>n}$ is monotone increasing in $x$ and $f$ is monotone increasing, we get 
\[
 \int_{\gO} g(\go)d\mu(\go)\,\int_\gO f(g(\go))d\mu(\go)\ \le\ \int_\gO g(\go) f(g(\go))d\mu(\go)
\]
and therefore, by the same monotonicity argument, 
\begin{equation}
\begin{aligned}
&\|a(\mu)-a_n(\mu)\|  \int_{\gO} g(\go)d\mu(\go)\,\int_\gO f(g(\go))d\mu(\go)\ \le\ C\int_{\gO\setminus X_n}(1+2g(\go))d\mu(\go)\,\int_\gO g(\go) f(g(\go))d\mu(\go)\\
&\le\ C\int_{\gO\setminus X_n}(1+2g(\go)) g(\go) f(g(\go))d\mu(\go)\,.
\end{aligned}
\end{equation}
Similarly, 
\begin{equation}
\begin{aligned}
&\int_{\gO\setminus X_n}g(\go)d\mu(\go)\,\int_\gO f(g(\go))d\mu(\go)\ =\ \int_{\gO}q(g(\go))d\mu(\go)\,\int_\gO f(g(\go))d\mu(\go)\\ 
&\le\ \int_\gO q(g(\go))f(g(\go))d\mu(\go)\ =\ \int_{\gO\setminus X_n}g(\go)f(g(\go))d\mu(\go)\,.
\end{aligned}
\end{equation}
Therefore, by the Fubini Theorem, 
\begin{equation}
\begin{aligned}
&|\int_{\Delta(\mu)} (\bar W(\mu)-\bar W(\mu_{X_n}))d\tau(\mu)|\\
& \le\ \int_{\gD(\gO)}\int_{\gO\setminus X_n}g(\go)f(g(\go))d\mu(\go)d\tau(\mu)\\
&+\  \int_{\gD(\gO)}C\int_{\gO\setminus X_n}(1+2g(\go)) g(\go) f(g(\go))d\mu(\go)d\tau(\mu)\\
&=\ \int_{\gO\setminus X_n}(g(\go)f(g(\go))+C(1+2g(\go)) g(\go) f(g(\go)))d\mu_0(\go)\,.
\end{aligned}
\end{equation}
Thus, Assumption \ref{integrability} implies that we can restrict our attention to the case when $\gO=X_n$ is compact. 

In this case, the Prokhorov Theorem implies that $\Delta(\gO)$ is compact in the weak* topology and this topology is metrizable. Thus, for any $\eps>0,$ we can decompose $\Delta(\gO)=Q_1\cup\cdots\cup Q_K$ where all $Q_k$ have diameters less than $\eps.$ We can now approximate $\tau$ by $\tilde\tau=\sum_k \nu_k \gd_{\mu_k}$ with $\mu_k=\int_{Q_k} \mu d\tau(\mu)/\nu_k$ and $\nu_k= \int_{Q_k} d\tau(\mu).$ Clearly, $\int \mu d\tilde\tau(\mu)=\mu_0,$ and therefore it remains to show that $\bar W$ is continuous in the weak* topology. 

To this end, suppose that $\mu_n\to\mu$ in the weak* topology. Let us first show $a_n=a(\mu_n)\to a(\mu)=a.$ Suppose the contrary.  Since $\gO$ is compact and $G$ is continuous and bounded, Lemma \ref{existence} implies that $a_n$ are uniformly bounded. Pick a subsequence such that $\|a_n-a\|>\eps$ for some $\eps>0$ and subsequence $a_n\to b$ for some $b\not=a.$ Since $G(a_n,\go)\to G(b,\go)$ uniformly on $\gO,$ we get a contradiction because 
\[
\int G(a_n,\go)d\mu_n-\int G(b,\go)d\mu\ =\ \int (G(a_n,\go)-G(b,\go))d\mu_n\ +\ \int G(b,\go)d(\mu_n-\mu)\,.
\]
The second term converges to zero because of weak* convergence. The first term can be bounded by 
\[
|\int (G(a_n,\go)-G(b,\go))d\mu_n|\ \le\ C \|a_n-b\|
\]
and hence also converges to zero. Thus, $\int G(b,\go)d\mu=\int G(a,\go)d\mu=0$, implying that $a=b$ by the strict monotonicity of the map $G.$ The same argument implies the required continuity of $\bar W(\mu).$ 
\end{proof}

\begin{proof}[Proof of Theorem \ref{mainth-limit}] For each finite $K,$ the optimal solution $(a_K(\go),x_K(a(\go)))$ stay uniformly bounded and hence there exists a subsequence converging in $L_2(\gO;\mu_0)$ and in probability to a limit $(a(\go),x(a(\go))).$ By continuity and Lemma \ref{lem-approx}, $E[W(a_K(\go),\go)]$ converges to the maximum in the problem of Definition \ref{def-u} and, hence, by the same continuity argument, $a(\go)$ is an  optimal policy without randomization. Since \eqref{olicy} holds true for $a_K$ for each finite $K$, convergence in probability implies that \eqref{x-top} also holds in the limit. Indeed, 
\[
c(a_K(q),\go;x_K)\ \ge\ c(a_K(\go),\go;x_K)
\]
holds for almost all $\go$ and all $q$ with probability one, and hence it also holds in the limit with probability one (due to convergence in probability). Clearly, for each finite $K$ the function $c(a_K(\go),\go;x_K)$ is smooth in each region $\gO_k$ and is continuous at the boundaries. Since $a,x$ stay bounded and $W,G$ are smooth and $G$ is compact, the functions are uniformly Lipschitz continuous and the Arzela-Ascoli theorem implies that so is the limit (passing to a subsequence if necessary). Finally, to prove that 
\[
E[G(a(\go),\go)|a(\go)]\ =\ 0
\]
it suffices to prove that 
\[
E[G(a(\go),\go)f(a(\go))]\ =\ 0
\]
for a countable dense set of test functions, which follows by passing to a subsequence. 

To verify all the required integrability to apply Lebesgue dominated convergence, Assumption \ref{integrability} implies that we just need to check that $E[D_aG(a,\go)|a]^{-1}$ is uniformly bounded. Since, by assumption, $\|D_aG(a,\go)\|$ is uniformly bounded, we just need to check that eigenvalues of $E[D_aG(a,\go)|a(\go)]$ are uniformly bounded away from zero. 

Indeed, let $\eps=\inf_{a,\go,z,\|z\|=1}-z\top D_aG(a,\go)z >0.$ If $\gl$ is an eigenvalue of $E[-D_aG(a,\go)|a]$ with a normed eigenvector $z$, then 
\[
\gl\ =\ z^\top E[-D_aG(a,\go)|a]z\ =\ E[-z^\top D_aG(a,\go)z|a]\ \ge\eps\,.
\]
To prove that \eqref{x-top} holds, we note that it suffices to show that 
\[
E[(x(a(\go))^\top E[D_aG(a,\go)|a(\go)]-E[D_aW(a,\go)|a(\go)])f(a(\go))]\ =\ 0
\]
for a countable, dense set of smooth test functions $f.$ The latter is equivalent to 
\[
E[(x(a(\go))^\top D_aG(a(\go),\go)-D_aW(a(\go),\go))f(a(\go))]\ =\ 0
\]
and the claim follows by continuity by passing to a subsequence. Finally, the last identity follows from \eqref{key}. Finally, the fact that $a$ is Borel-measurable follows from the known fact that for any Lebesgue-measurable $a(\go)$ there exists a Borel measurable modification of $a$ coinciding with $a$ for Lebesgue-almost every $\go.$ 
\end{proof}

\begin{proof}[Proof of Corollary \ref{cor-moment}] Integrability condition (by the same argument as in the proof of Lemma \ref{lem-approx}) implies that all the convergence arguments are justified. The convexity claim is then a direct consequence of Proposition \ref{main convexity}. 

\end{proof}

\begin{proposition}\label{delta-dev} Let $\gamma$ be the joint distribution of $(a,\go)$ for an optimal information design. Then, 
\[
\int (x(a)^\top G(a,\go)- W(a,\go))d\eta\ +\ \int W(\tilde a_*,\go)d\eta(\cR,\go)\ \le\ 0
\]
for every measure $\eta$ such that $\Supp(\eta)\subset\Supp(\gamma)\,$ such that $\int f(\|a\|^2)g^2(\go)d\eta(a,\go) <\infty.$ 
In the case of moment persuasion, 
\begin{equation}\label{kramkov2019optimal-12}
\int (D_aW(a)(a -g(\go))-W(a))d\eta\ +\ W(\int g(\go)d\eta)\le 0
\end{equation}
\end{proposition}

\begin{proof}[Proof of Proposition \ref{delta-dev}] We closely follow the arguments and notation in \cite{kramkov2019optimal}. Let $\gamma$ be the joint distribution of the random variables $\go$ and $a(\go).$ We first establish \eqref{kramkov2019optimal-12} for a Borel probability measure $\eta$ that has a bounded density with respect to $\gamma.$ Then, the general result follows by a simple modification of the argument in the proof of Theorem A.1 in \cite{kramkov2019optimal}. Let
\[
V(a,\go)\ =\ \frac{d\eta}{d\gamma}(a,\go)\,.
\]
We choose a non-atom $q\in \cR$ of $\mu(da)\ =\ \gamma(da,\R^L)$ and define the probability measure
\[
\zeta(da,d\go)\ =\ \gd_q(da)\eta(\cR,d\go)\,,
\]
where $\gd_q$ is the Dirac measure concentrated at $q.$ For sufficiently small $\eps>0$ the probability measure
\[
\tilde\gamma\ =\ \gamma\ +\ \eps (\zeta-\eta)
\]
is well-defined and has the same $\go$-marginal $\mu_0(\go)$ as $\gamma$. Let $\tilde a$ be the optimal action satisfying 
\[
\tilde\gamma (G(\tilde a,\go)|\tilde a)\ =\ 0\,.
\]
The optimality of $\gamma$ implies that 
\begin{equation}\label{kramkov2019optimal-13}
\int W(\tilde a,\go)d\tilde\gamma\ \le\ \int W(a,\go)d\gamma\,. 
\end{equation}
By direct calculation, 
\begin{equation}
\begin{aligned}
&0\ =\ \tilde\gamma (G(\tilde a,\go)|a)\\
&=\ {\bf 1}_{a\not=q}\frac{\int G(\tilde a,\go)d((\gamma|a)-\eps (\eta|a))}{\int d(\gamma-\eps \eta)}\ +\ {\bf 1}_{a=q}\int G(\tilde a,\go) d\eta(\cR,\go)\\
&=\ {\bf 1}_{a\not=q}\frac{\int G(\tilde a,\go)d(\gamma|a)-\eps\int G(\tilde a,\go)d (\eta|a)}{1-\eps U(a)}\ +\ {\bf 1}_{a=q}\int G(\tilde a,\go) d\eta(\cR,\go)
\end{aligned}
\end{equation}
where $U(a)=\gamma(V(a,\go)|a)\,.$
Now, we know that 
\[
\int G(a,\go)d(\gamma|a)\ =\ 0,
\]
and the assumed regularity of $G$ together with the implicit function theorem imply that 
\[
\tilde a(a)\ =\ a\ +\ \eps Q(a)\ +\ O(\eps^2)
\]
if $a\not=q$ and
\[
\tilde a\ =\ \tilde a_*\,,
\]
where $\tilde a_*$ is the unique solution to 
\[
\int G(\tilde a_*,\go) d\eta(\cR,\go)\ =\ 0
\]
for $a=q.$ Here, 
\begin{equation}
\begin{aligned}
&0\ =\ O(\eps^2)\ +\ \int G(a\ +\ \eps Q(a),\go)d(\gamma|a)-\eps\int G(a,\go) V(a,\go) d (\gamma|a)\\
&=\ O(\eps^2)\ +\ \eps \int D_aG(a,\go) d(\gamma|a) Q(a)-\eps\int G(a,\go) V(a,\go) d (\gamma|a)
\end{aligned}
\end{equation}
so that 
\[
Q(a)\ =\ \left(\int D_aG(a,\go) d(\gamma|a)\right)^{-1}\int G(a,\go) V(a,\go) d (\gamma|a)\,. 
\]
Thus, 
\begin{equation}
\begin{aligned}
&\int W(\tilde a(a),\go)d\tilde \gamma\ =\ \int W(\tilde a(a),\go)(1-\eps V(a,\go))d\gamma+\eps \int W(\tilde a_*,\go)d\eta(\cR,\go)\\
&=\ O(\eps^2)\ +\ \int W(a,\go)d\gamma\ +\ \eps\Bigg(\int (D_aW(a,\go) Q(a) - V(a,\go)) d\gamma\ +\ \int W(\tilde a_*,\go)d\eta(\cR,\go)
\Bigg)
\end{aligned}
\end{equation}
In view of \eqref{kramkov2019optimal-13}, the first-order term is non-positive:
\[
\int (D_aW(a,\go) Q(a) - W(a,\go)V(a,\go)) d\gamma\ +\ \int W(\tilde a_*,\go)d\eta(\cR,\go)\ \le\ 0\,.
\]
Substituting, we get 
\[
\int (x(a)^\top \int G(a,\go) V(a,\go) d (\gamma|a)\ - W(a,\go)V(a,\go)) d\gamma\ +\ \int W(\tilde a_*,\go)d\eta(\cR,\go)\ \le\ 0\,,
\]
which is equivalent to 
\[
\int (x(a)^\top G(a,\go)- W(a,\go))d\eta\ +\ \int W(\tilde a_*,\go)d\eta(\cR,\go)\ \le\ 0
\]
In the case of a moment persuasion, we get 
\[
Q(a)\ =\ a U(a)\ -\ R(a)\,,
\]
where we have defined 
\[
U(a)\ =\ \gamma(V(a,\go)|a),\ R(a)\ =\ \gamma(g(\go)V(a,\go)|a)\,,
\]
and 
\[
\tilde a_*\ =\ \int g(\go)d\eta\,.
\]
Thus, we get 
\begin{equation}
\begin{aligned}
&0\ \ge\ \int (D_aW(a) Q(a) - W(a)V(a,\go)) d\gamma\ +\ W(\tilde a_*)\\
&=\ \int (D_aW(a)(a U(a)\ -\ R(a))-W(a)V(a,\go))d\gamma\ +\ W(\tilde a_*)\\
&=\ \int (D_aW(a)(a -g(\go))-W(a))d\eta\ +\ W(\int g(\go)d\eta)\,. 
\end{aligned}
\end{equation}
\end{proof}

\begin{lemma}\label{super-G} Let $a_*(\go)$ be the unique solution to $G(a_*(\go),\go)=0.$ Then, at any optimal $(a(\go),\,x(a(\go))$ with $x(a(\go))=\bar D_aW(a(\go))\bar D_aG(a(\go))^{-1}$ we have 
\[
x(a(\go))^\top G(a(\go),\go)\ -\ W(a(\go),\go)\ +\ W(a_*(\go),\go)\ \le\ 0
\]
Furthermore, 
\[
c(a(\go),\go;x)\ =\ \min_{a\in\Xi}\,c(a,\go;x)\,
\]
almost surely.
\end{lemma}

\begin{proof} The first claim follows by selecting $\eta=\gd_{(a(\go),\go)}.$ The second one follows by selecting $\eta=t\gd_{a_1}{\bf 1}_{\gO_1}\gamma|a_1+(1-\kappa t)\gd_{a_2}\gamma|a_2$ for some open set $\gO_1$ and $\kappa=\gamma(\gO_1|a_1)$. In this case, we get 
\begin{equation}\label{aux10}
\begin{aligned}
&t\int (x(a_1)^\top G(a_1,\go)- W(a_1,\go)){\bf 1}_{\gO_1}d\gamma(\go|a_1)\ +\ (1-\kappa t)\int (x(a_2)^\top G(a_2,\go)- W(a_1,\go))d\gamma(\go|a_2)\\
& +\ \int W(\tilde a_*,\go)(t{\bf 1}_{\gO_1}d\gamma|a_1+(1-\kappa t)d\gamma|a_2)\ \le\ 0\,. 
\end{aligned}
\end{equation}
where $\tilde a_*(t)$ is uniquely determined by 
\[
t\int G(\tilde a_*(t),\go){\bf 1}_{\gO_1}d(\gamma|a_1)+(1-\kappa t)\int G(\tilde a_*(t),\go)d(\gamma|a_2)\ =\ 0\,.
\]
Clearly, \eqref{aux10} is equivalent to 
\begin{equation}\label{aux11111}
\begin{aligned}
&t\int (W(\tilde a_*(t),\go)-W(a_1,\go)+x(a_1)^\top G(a_1,\go)){\bf 1}_{\gO_1} d(\gamma|a_1)\\
&+(1-\kappa t)\int (W(\tilde a_*(t),\go)-W(a_2,\go)+x(a_2)^\top G(a_2,\go))d(\gamma|a_2)\ \le\ 0\,. 
\end{aligned}
\end{equation}
Assuming that $t$ is small, we get 
\[
\tilde a_*(t)\ =\ a_2+t \hat a\ +\ o(t),\ \hat a\ =\ -\bar D_aG(a_2)^{-1}\int G(a_2,\go){\bf 1}_{\gO_1}d(\gamma|a_1)
\]
and hence
\begin{equation}\label{aux11}
\begin{aligned}
&0\ge t\int (W(\tilde a_*(t),\go)-W(a_1,\go)+x(a_1)^\top G(a_1,\go)){\bf 1}_{\gO_1} d(\gamma|a_1)\\
&+(1-\kappa t)\int (W(\tilde a_*(t),\go)-W(a_2,\go)+x(a_2)^\top G(a_2,\go))d(\gamma|a_2)\\
&=\ t\int (W(a_2,\go)-W(a_1,\go)+x(a_1)^\top G(a_1,\go)){\bf 1}_{\gO_1} d(\gamma|a_1)\\
&+t\bar D_aW(a_2)\hat a\ +\ o(t)\\
&=\ t\int (W(a_2,\go)-W(a_1,\go)+x(a_1)^\top G(a_1,\go)){\bf 1}_{\gO_1} d(\gamma|a_1)\\
&-t\bar D_aW(a_2)\bar D_aG(a_2)^{-1}\int G(a_2,\go){\bf 1}_{\gO_1}d(\gamma|a_1)\ +\ o(t)\\
&=\ t\int(c(a_1,\go;x)-c(a_2,\go;x)){\bf 1}_{\gO_1}d(\gamma|a_1)\ +\ o(t)\,.
\end{aligned}
\end{equation}
Since $\gO_1$ is arbitrary, we get that 
\[
c(a_1,\go;x)\ \le\ c(a_2,\go;x)
\]
almost surely with respect to $\gamma|a_1.$ 
\end{proof}

\begin{corollary}\label{main-conditions} In a moment persuasion setup, let $a(\go)$ be an optimal information design. Then, for any two points $\go_1,\go_2,$ we have 
\begin{equation}
\begin{aligned}
&W(t g(\go_1)+(1-t)g(\go_2))\ +\ t (D_aW(a(\go_1))(a(\go_1) -g(\go_1))-W(a(\go_1)))\\ 
&+\ (1-t)(D_aW(a(\go_2))(a(\go_2) -g(\go_2))-W(a(\go_2)))\ \le\ 0\,. 
\end{aligned}
\end{equation}
In the case of $\go_1=\go_2,$ we just get 
\begin{equation}\label{simple-max}
 D_aW(a(\go))(a(\go) -g(\go))\ \le\ W(a(\go))\ -\ W(g(\go))\,.
\end{equation}
Furthermore, 
\begin{equation}\label{conv-W}
W(t a_1+(1-t)a_2)\ \le\ tW(a_1)\ +\ (1-t) W(a_2)\ ,\ a_1,\ a_2\in \Supp(a)\,. 
\end{equation}
In particular, $\Supp(a)$ is a $W$-convex set. 
\end{corollary}

\begin{proof}[Proof of Corollary \ref{main-conditions}] The first claim follows from the choice $\eta=t\gd_{(a(\go_1),\go_1)}+(1-t)\gd_{(a(\go_2),\go_2)}$ in Proposition \ref{delta-dev}. The second one follows from Lemma \ref{super-G}. Monotonicity of the set follows by evaluating the inequality at $t\to 0.$ Maximality follows from \eqref{simple-max}.
\end{proof}

\begin{proof}[Proof of Theorem \ref{monge}] Let $(a(\go),\ x(a))$ be an optimal policy and let 
\[
\phi^c(a)\ =\ \inf_\go (c(a,\go;x)-\phi_\Xi(\go;x))
\]
Pick an $a\in \Xi.$ Since $a\in \Xi,$ there exists a $\tilde\go$ such that $a=a(\tilde\go)$ and hence 
\[
\phi^c(a)\ =\ \inf_\go (c(a,\go;x)-\phi_\Xi(\go;x))\ \le\ c(a,\tilde\go;x)-\phi_\Xi(\go;x)\ =\ 0\,.
\]
Thus, 
\[
\int\phi^c(a(\go))\mu_0(\go)d\go\ =\ 0\,. 
\]
At the same time, 
\[
c(a,\go;x)-\phi_\Xi(\go;x)\ =\ c(a,\go;x)\ -\ \inf_{b\in \Xi}c(b,\go;x)\ \ge 0\,.
\]
Thus, $\phi^c(a)=0$ for all $a\in \Xi.$ Now, by the definition of $\phi^c,$ we always have 
\[
\phi^c(a)\ +\ \phi_\Xi(\go;x)\ \le\ c(a,\go)\
\]
for an optimal policy. Let $\gamma$ be the measure on $\Xi\times \gO$ describing the joint distribution of $\chi=a(\go)$ and $\go.$ Then, 
\[
\int c(a,\go)\gamma(a,\go)\ =\ \int c(a(\go),\go)\mu_0(\go)d\go\ =\ \int \phi_\Xi(\go;x)\mu_0(\go)d\go\ =\ \int \phi_\Xi(a)d\nu(a)\,,
\]
Pick any measure $\pi$ from the Kantorovich problem. Then, 
\begin{equation}
\begin{aligned}
&\int c(a,\go)d\gamma(a,\go)\ =\ \int(\phi_\Xi(\go;x))d\gamma(a,\go)\\
& =\  \int\phi_\Xi(\go;x)\mu_0(\go)d\go\ +\ \int\phi^c(a(\go))\mu_0(\go)d\go\\
&=\ \int(\phi_\Xi(\go;x)+\phi^c(a(\go)))\mu_0(\go)d\go\\
&=\ \int(\phi_\Xi(\go;x)+\phi^c(a))d\pi(a,\go)\ \le\ \int c(a,\go)d\pi(a,\go)
\end{aligned}
\end{equation}
Thus, $\gamma$ minimizes the cost in the Kantorovich problem. 
\end{proof}

\begin{proof}[Proof of Theorem \ref{converse}] The first claim follows from Corollary \ref{main-conditions} in the Appendix. The proof of sufficiency closely follows ideas from \cite{kramkov2019optimal}. 

Let $a(\go)$ be a policy satisfying the conditions Theorem \ref{converse}. Note that, in terms of the function $c,$ our objective is to show (see \eqref{key}) that 
\[
\min_{all\ feasible\ policies\ b(\go)}E[c(b(\go),g(\go))]\ =\ E[c(a(\go),g(\go))]. 
\]
Next, we note that the assumed maximality implies that $c(a(\go),g(\go))\ =\ \phi_\Xi(g(\go))\le 0$ for all $\go.$ Now, for any feasible policy $b(\go),$ we have $E[g(\go)|b(\go)]=b(\go)\in conv(g(\gO))$ and therefore, for any fixed $a\in \R^M,$ we have 
\begin{equation}
\begin{aligned}
&E[c(a,g(\go))\ -\ c(b(\go),g(\go))|b(\go)]\\
& =\ E[W(g(\go))\ -\ W(a)\ +\ D_aW(a)\,(a-g(\go))\\
&-(W(g(\go))\ -\ W(b(\go))\ +\ D_aW(b(\go))\,(b(\go)-g(\go)))]\\
&\ =\ W(b(\go))-W(a)+D_aW(a)(a-b(\go))\ =\ c(a,b(\go))\,.
\end{aligned}
\end{equation}
Taking the infinum over a dense, countable set of $a,$ we get 
\[
\inf_a E[c(a,g(\go))\ -\ c(b(\go),g(\go))|b(\go)]\ =\ \phi_\Xi(b(\go))\le 0
\]
and therefore 
\begin{equation}\label{last-c}
\begin{aligned}
&E[c(a(\go),g(\go))\ -\ c(b(\go),g(\go))|b(\go)]\ =\ E[\inf_{a\in \Xi}c(a,g(\go))\ -\ c(b(\go),g(\go))|b(\go)]\\ 
&\le\ \inf_{a\in\Xi}E[c(a,g(\go))\ -\ c(b(\go),g(\go))|b(\go)]= \inf_{a\in \Xi} c(a,b(\go))\le 0\,.
\end{aligned}
\end{equation}
and therefore, integrating over $b,$ we get  
\[
E[c(a(\go),g(\go))]\ \le\  E[c(b(\go),g(\go))]\,. 
\]
The proof is complete. 
\end{proof}

\begin{proof}[Proof of Proposition \ref{uniqueness}] Since $\Xi$ is $W$-monotone, we have $c(a,b)\ge 0$ for all $a,b\in \Xi$ and hence $\phi_\Xi(b)\ge 0$ for all $b\in \Xi.$ Thus, if $\Xi\subset X,$ we get $\phi_\Xi(b)=0$ on $\Xi.$ Let now $a$ be an optimal policy. First we note that \eqref{last-c} implies that $\phi_\Xi(b(\go))=0$ almost surely for any optimal policy $b(\go).$ 

If $\Xi=Q_\Xi,$ we get that $\tilde\Xi\subset \Xi$ and hence $\phi_\Xi(b) \le \phi_{\tilde\Xi}(b)$ for all $b.$  Thus, 
\[
\int c(b(\go),g(\go))\mu_0(\go)d\go\ \ge\ \int\phi_{\tilde\Xi}(g(\go))\mu_0(\go)d\go\ \ge \int\phi_{\Xi}(g(\go))\mu_0(\go)d\go\,.
\]
Since both policies are optimal, we must have $\phi_\Xi=\phi_{\tilde\Xi},$ and the singleton assumption implies that $a(\go)=\tilde a(\go).$ 
\end{proof}

\begin{proof}[Proof of Theorem \ref{frostman}] Our proof is based on an application of the famous Frostman's lemma (see, e.g., \cite{mattila1999geometry}). 

\begin{lemma}[Frostman's lemma] We only consider the case of a non-degenerate $D_{aa}W.$ The case of a degenerate, constant $D_{aa}W$ is proved analogously. 

Define the $s$-capacity of a Borel set $A$ as follows:
\[
C_s(A)\ =\ \sup\left\{
\left(
\int_{A\times A}\frac{d\mu(x)d\mu(y)}{\|x-y\|^s}
\right)^{-1}:\ \mu\ \text{is a Borel measure and $\mu(A)=1$}
\right\}\,.
\]
(Here, we take $\inf\emptyset=\infty$ and $1/\infty=0.$) Then, 
\[
\dim_H(A)\ =\ \sup\{s\ge 0:\ C_s(A)\ >\ 0\}\,. 
\]
\end{lemma}
Now, by Corollary \ref{main-conditions} (formula \eqref{conv-W}), we have that the function $q(t)=W(a_1 t+a_2(1-t)),\ t\in (0,1)$ is either identically constant or attains a global minimum at a point $t_*\in (0,1).$ At this point, we have 
\[
(a_1-a_2)^\top D_{aa}W(t_*a_1+(1-t_*)a_2)(a_1-a_2)\ \ge\ 0. 
\]
Without loss of generality, we may assume that $X$ is a sufficiently small ball in $\R^M.$ Pick any point $a_*\in X$ for which $D_{aa}W$ is non-degenerate and change the coordinates so that $D_{aa}W(a_*)=\diag(\gl_1,\cdots,\gl_M)$ is diagonal. Furthermore, rescaling the coordinates, we may assume that all $\gl_i$ have absolute values equal to $1$, so that $\gl_i=1$ for $i\le \nu(a_*)$ and $\gl_i=-1$ for $i>\nu(a_*)$. Furthermore, making the ball $X$ sufficiently small, we may assume that $\|D_{aa}W(a)-D_{aa}W(a_*)\|\le \eps$ for all $a\in X.$ Let $a_i=(x_i,y_i)$ be the orthogonal decomposition into two components corresponding to positive and negative eigenvalues. Then, 
\begin{equation}
\begin{aligned}
&0\ \le\ (a_1-a_2)^\top D_{aa}W(t_*a_1+(1-t_*)a_2)(a_1-a_2)\\
& \le\ \eps\|a_1-a_2\|^2\ +\ \|x_1-x_2\|^2-\|y_1-y_2\|^2\ =\ (1+\eps)\|x_1-x_2\|^2-(1-\eps) \|y_1-y_2\|^2\,,
\end{aligned}
\end{equation}
and therefore 
\[
\|a_1-a_2\|^2\ =\  \|x_1-x_2\|^2+\|y_1-y_2\|^2\ \le\ c^2 \|x_1-x_2\|^2
\]
with $c^2=1+(1+\eps)/(1-\eps).$ 
\begin{equation}
\begin{aligned}
&\int_{\Xi\cap X\times \Xi\cap X}\frac{d\mu(a_1)d\mu(a_2)}{\|a_1-a_2\|^s}\\ 
&\ge c^{-1}\int_{\Xi\cap X\times \Xi\cap X}\frac{d\mu(a_1)d\mu(a_2)}{\|x_1-x_2\|^s}\ =\ c^{-1}\int_{\tilde \Xi\cap X\times \tilde \Xi\cap X}\frac{d\tilde\mu(x_1)d\tilde\mu(x_2)}{\|x_1-x_2\|^s}
\end{aligned}
\end{equation}
where $\tilde\mu$ is the $x$-marginal of the measure $\mu$ and $\tilde\Xi$ is the projection of $\Xi$ onto $\R^{\nu(a_*)}.$ Thus, 
\[
C_s(\Xi\cap X)\ \le\ c\,C_s(\tilde\Xi\cap X)
\]
and therefore
\[
\dim_H(\Xi\cap X)\ \le\ \dim(\tilde\Xi\cap X)\,.
\]
Since $\tilde\Xi\subset\R^{\nu(a_*)},$ the claim follows. 
\end{proof}

\begin{proof}[Proof of Corollary \ref{cor-dirichlet}]  We will use an important property of the Dirichlet distribution: Defining $\bar\go_1=(\go_1,\cdots,\go_j)$ and $\bar\go_2=(\go_{j+1},\cdots,\go_M),$ we have that the random vectors ${\bf 1}^T\bar\go_1,\ \bar\go_1/({\bf 1}^T\bar\go_1)$ and $\bar\go_2/(1-{\bf 1}^\top \bar\go_1)$ are jointly independent. 
In this case, by direct calculation, 
\[
a(\go)\ =\ E[\binom{\psi_1 \bar\go_1}{\psi_2 \bar\go_2}|a(\go)]\,. 
\]
Then, 
\begin{equation}
\begin{aligned}
&\phi_\Xi(b)\ =\ \min_{\go}\Bigg(W(b)-W(a(\go))+D_aW(a(\go))(a(\go)-b)\Bigg)\\
&=\ \min_{\go}\Bigg(W(b)-\sum_i (q_i\cE_i(a_i(\go))+\varphi_i({\bf 1}^\top a_i(\go)))\\
&+\sum_i (q_i(\log(a_i/y_i)+1)+\varphi_i'({\bf 1}^\top a_i))^\top(a_i(\go)-b_i)\Bigg)\\
\end{aligned}
\end{equation}
Since ${\bf 1}^\top a_i(\go)=\gamma_i,$ the first order conditions take the form 
\[
q_1a_1^{-1}(a_1-b_1)=\gl_1,\ q_2a_2^{-1}(a_2-b_2)=\gl_2
\]
where $\gl_i$ are Lagrange multipliers for the constraints ${\bf 1}^\top a_i=\gamma_i\,.$ Thus, denoting $\bar b_i={\bf 1}^\top b_i,$ we get $a_i=\gamma_i b_i/\bar b_i,$ and using the identity 
\[
\cE(\frac{\gamma b_i}{\bar b})\ =\ \frac{\gamma}{\bar b}(\cE(b)+\log(\frac{\gamma}{\bar b})\bar b)
\]
we get 
 \begin{equation}
\begin{aligned}
&\phi_\Xi(b)\ =\ \sum_{i=1}^2\left(\varphi_i(\bar b_i)-\varphi_i(\gamma_i)+\varphi_i'(\gamma_i)(\gamma-\bar b_i)\ -\ q_i(\bar b_i \log(\gamma_i/\bar b_i)-\gamma_i+\bar b_i)\right)
\end{aligned}
\end{equation}

\end{proof}

\newpage

\bibliographystyle{aer}
\bibliography{bibliography}

\newpage

\appendix

\end{document}